\documentclass[sigconf]{acmart}

\AtBeginDocument{%
  \providecommand\BibTeX{{%
    \normalfont B\kern-0.5em{\scshape i\kern-0.25em b}\kern-0.8em\TeX}}}

\usepackage{hyperref}
\usepackage{url}
 \usepackage[ruled]{algorithm2e}

\usepackage{graphicx}
\usepackage{capt-of}
\usepackage{resizegather}
\usepackage{caption}
\usepackage{float}

\usepackage{amssymb}
\usepackage{subfigure}
\usepackage{threeparttable}
\usepackage{wrapfig}       
\usepackage{diagbox}
\usepackage{enumitem}
\usepackage{multirow}
\usepackage{makecell}
\usepackage{booktabs}       
\usepackage{pifont}

\usepackage{amsmath}
\usepackage{amsthm}
\def \P {\mathbb{P}}

\def \bfE {\mathbb{E}}

\def \cL{\mathcal{L}}

\def \cE{\mathcal{E}}
 \def \cD{\mathcal{D}}

\copyrightyear{2024}
\acmYear{2024}
\setcopyright{acmlicensed}
\acmConference[KDD '24] {Proceedings of the 30th ACM SIGKDD Conference on Knowledge Discovery and Data Mining }{August 25--29, 2024}{Barcelona, Spain.}
\acmBooktitle{Proceedings of the 30th ACM SIGKDD Conference on Knowledge Discovery and Data Mining (KDD '24), August 25--29, 2024, Barcelona, Spain}
\acmDOI{10.1145/3637528.3671915}
\acmISBN{979-8-4007-0490-1/24/08}


 \author{Haoxuan Li}
\affiliation{
  \institution{Peking University}
  \country{Beijing, China}}
 \email{hxli@stu.pku.edu.cn}

 \author{Chunyuan Zheng}
\affiliation{
  \institution{{Peking University}}
  \country{Beijing, China}}
 \email{zhengchunyuan99@gmail.com}

 \author{Wenjie Wang}
\affiliation{
  \institution{National University of Singapore}
  \country{Singapore, Singapore}}
 \email{wenjiewang96@gmail.com}

\author{Hao Wang}
\affiliation{
  \institution{Zhejiang University}
  \country{Hangzhou, China}}
 \email{haohaow@zju.edu.cn}

\author{Fuli Feng}
\affiliation{
  \institution{University of Science and\\Technology of China}
  \country{}
  }
 \email{fulifeng93@gmail.com}

 \author{Xiao-Hua Zhou}
 \authornote{Corresponding author}
\affiliation{
  \institution{Peking University}
  \country{Beijing, China}}
 \email{azhou@math.pku.edu.cn}

\begin{document}

\title{Debiased Recommendation with Noisy Feedback}



\begin{abstract}

Ratings of a user to most items in recommender systems are usually missing not at random (MNAR), largely because users are free to choose which items to rate. To achieve unbiased learning of the prediction model under MNAR data, three typical solutions have been proposed, including error-imputation-based (EIB), inverse-propensity-scoring (IPS), and doubly robust (DR) methods. However, these methods ignore an alternative form of bias caused by the inconsistency between the observed ratings and the users' true preferences, also known as noisy feedback or outcome measurement errors (OME), \emph{e.g.}, due to public opinion or low-quality data collection process. In this work, we study intersectional threats to the unbiased learning of the prediction model from data MNAR and OME in the collected data. First, we design OME-EIB, OME-IPS, and OME-DR estimators, which largely extend the existing estimators to combat OME in real-world recommendation scenarios. Next, we theoretically prove the unbiasedness and generalization bound of the proposed estimators. We further propose an alternate denoising training approach to achieve unbiased learning of the prediction model under MNAR data with OME. Extensive experiments are conducted on three real-world datasets and one semi-synthetic dataset to show the effectiveness of our proposed approaches. The code is available at \href{https://github.com/haoxuanli-pku/KDD24-OME-DR}{https://github.com/haoxuanli-pku/KDD24-OME-DR}.
\end{abstract}

\maketitle

\section{Introduction}
Recommender systems (RS) are designed to generate meaningful recommendations to a collection of users for items or products that might interest them, which have made a number of significant advancements in recent years~\cite{zhao2019deep,zhang2020explainable,rui2022knowledge,Zhang-etal-2023}. Nevertheless, direct use of these advanced models in real-world scenarios while ignoring the presence of numerous biases in the collected data can lead to sub-optimal performance of the rating prediction model~\cite{Chen-etal2022,yang2021top,yang2023debiased}. Among these, the problem of missing data is particularly prevalent, as users are free to choose items to rate and the ratings with a lower value are more likely to be missing~\cite{de2014reducing}, leading to the collected data in RS is always missing not at random (MNAR)~\cite{marlin2009collaborative}. MNAR ratings pose a serious challenge to the unbiased evaluation and learning of the prediction model because the observed data might not faithfully
represent the entirety of user-item pairs~\cite{Schnabel-Swaminathan2016,Wang-Zhang-Sun-Qi2019}. 

To tackle this problem, previous studies evaluate the performance of a prediction model by computing the \emph{prediction inaccuracy}: the average of the prediction errors (\emph{e.g.}, the squared difference between a
predicted rating and the \emph{potentially observed} rating) for all ratings~\cite{saito2019unbiased,saito2020doubly}. To unbiasedly estimate the prediction inaccuracy when the ratings are partially observable, three typical approaches have been proposed, including: (1) The error-imputation-based (EIB) approaches~\cite{Steck2010,Lobato-etal2014}, which compute an imputed error for each missing rating. (2) The inverse-propensity-scoring (IPS) approaches~\cite{Schnabel-Swaminathan2016,saito2020ips}, which inversely weight the prediction error for each observed rating with the probability of observing that rating. (3) The doubly robust (DR) approaches~\cite{SDR,song2023cdr, li2024kernel,li2024uidr}, which use both the error imputation model and the propensity model to estimate the prediction inaccuracy, and the estimation is unbiased when either the imputed errors or the learned propensities are accurate.

Despite the widespread use of these methods, they ignore an alternative form of bias caused by the inconsistency between the observed ratings and the users' true preferences, also known as noisy feedback or outcome measurement errors (OME). Similar to selection bias, OME also arises from systematic
bias during the data collection. For example, the collected user feedback may differ from the true user preferences due to the influence of public opinions~\cite{zheng2021-distangle}. Meanwhile, low-quality data collection such as recommender attacks~\cite{gunes2014shilling} or carelessly filling out the after-sales assessment can also result in noisy user feedback~\cite{wang2021denoising}. Therefore, to make current debiased recommendation techniques more applicable to real-world scenarios that may contain OME, it is worthwhile to develop new model evaluation criteria and corresponding estimators to achieve unbiased learning under OME.

In this work, we study intersectional threats to the unbiased learning of the prediction model introduced by data MNAR and OME from historical interactions, where the observed ratings are not only MNAR but may differ from the true ratings due to the presence of OME. In such a context, a natural evaluation criterion for the prediction models is proposed, called \emph{true prediction inaccuracy}: the average of the \emph{true} prediction errors (\emph{e.g.}, the squared difference between a
predicted rating and the \emph{true} rating) for all ratings. Given the collected data with OME, the true prediction errors are difficult to obtain even for the observed ratings, making the previous debiasing estimators severely biased in estimating the true prediction inaccuracy and training the rating prediction model.

To combat the influence of OME on the performance of the prediction model, many data-driven error parameter estimation methods are developed in recent machine learning literature~\cite{northcutt2021confident,scott2013classification,scott2015rate,xia2019anchor}. Meanwhile, given knowledge
of measurement error parameters, recent studies propose unbiased risk minimization approaches for learning under noisy labels~\cite{natarajan2013learning,chou2020unbiased,patrini2017making,van2015machine}. Despite the prevalence of OME in real-world recommendation scenarios, there is still only limited work focusing on denoising in RS~\cite{qin2021world,chen2022denoising,liu2021concept,zhang2023denoising}. 
For example, \cite{zhang2023robust} proposes to overcome noisy confounders by leveraging the sensitivity analysis in the statistics literature.
By noticing that noisy feedback typically has large loss values in the early stages, \cite{wang2021denoising} proposes adaptive denoising training and \cite{gao2022self} proposes self-guided denoising learning for implicit feedback. However, most of these methods are heuristic and lack theoretical guarantees of statistical unbiasedness for estimating the true prediction inaccuracy.

To address the above problem, we extend the widely adopted EIB, IPS, and DR estimators to achieve unbiased estimations with the true prediction inaccuracy under OME, named OME-EIB, OME-IPS, and OME-DR estimators, respectively. Our methods are built upon the existing weak separability assumption~\cite{menon2015learning}, which states that there exist "perfectly positive" and "perfectly negative" samples among the entire user-item pairs. Specifically, we discuss the rationality of the weak separability assumption in real-world recommendation scenarios and build up the linkage between the observed and true outcome probabilities at specific instances, from which we are able to obtain unbiased estimations of the measurement error parameters. We also derive explicit forms of the biases and the generalization bounds of the proposed OME-EIB, OME-IPS, and OME-DR estimators under inaccurately estimated measurement error parameters, from which we prove the double robustness of the OME-DR estimator. We further propose an alternating denoise training approach to achieve unbiased learning of the prediction model, which corrects for data MNAR and OME in parallel. In particular, the imputation model learns to accurately estimate the prediction errors made by the prediction model, while the prediction model learns from the imputation model to reduce the prediction errors in itself. In this way, the prediction and imputation models mutually regularize each other to reduce both prediction and imputation inaccuracies. The effectiveness of our proposed approaches is validated on three real-world datasets and one semi-synthetic dataset with varying MNAR levels and OME rates. To the best of our knowledge, our holistic evaluation is the first to examine how OME with its measurement error parameters' estimation in the context of selection bias interact to affect the debiased recommendations.

The contributions of this paper are summarized as follows.
    \vspace{-0.6mm}
\begin{itemize}[leftmargin=*]
    \item We formulate OME caused by the inconsistency between the observed ratings and the users' true preferences, and establish an evaluation criterion for the unbiased learning of the prediction model under OME, named true prediction inaccuracy.
    \item We develop OME-EIB, OME-IPS, and OME-DR estimators for unbiasedly estimating the true prediction inaccuracy of the prediction model  under MNAR data with OME, and theoretically
analyze the biases and generalization bounds of the estimators.
    \item We further propose an alternating denoise training approach to estimate the measurement error parameters and achieve unbiased learning of the prediction model under MNAR data with OME, which corrects for data MNAR and OME in parallel.
    \item We conduct extensive experiments on three real-world datasets and one semi-synthetic dataset, and the results demonstrate the superiority of our methods with varying MNAR and OME rates.
\end{itemize}
\section{Preliminaries}
\subsection{Task Formulation without OME}
Let $\mathcal{U}=\{u_1, \ldots, u_N\}$ be a set of users, $\mathcal{I}=\{i_1, \ldots, i_M\}$ a set of items, and $\mathcal{D}=\mathcal{U} \times \mathcal{I}$ the collection of all user-item pairs. The potentially observed rating matrix $\mathbf{R} \in \mathbb{R}^{N \times M}$ comprises potentially observed ratings $r_{u, i}$, which represents the observed rating if user $u$ had rated the item $i$. In RS, given the user-item features $x_{u, i}$, the prediction model $f_{\theta}(x_{u,i})$ parameterized by $\theta$ aims to accurately predict the ratings of all users for all items, and then recommend to the user the items with the highest predicted ratings, as described in \cite{ricci2010introduction}. The prediction matrix $\hat{\mathbf{R}} \in \mathbb{R}^{N \times M}$ comprises the predicted ratings $f_{\theta}(x_{u,i})$ obtained from this prediction model. If the rating matrix $\mathbf{R}$ had been fully observed, then the \emph{prediction inaccuracy} $\mathcal{P}$ of the prediction model can be measured using
$$
\mathcal{P}=\mathcal{P}(\hat{\mathbf{R}}, \mathbf{R}) = \frac{1}{|\cD|} \sum_{(u,i)\in \cD} e_{u,i},
$$
where $e_{u,i} = \ell(f_{\theta}(x_{u,i}), r_{u,i})$ is the prediction error, and $\ell(\cdot, \cdot)$ is a pre-defined loss, such as the mean square error (MSE), \emph{i.e.}, $e_{u, i}=(f_{\theta}(x_{u,i})-r_{u, i})^2$. Let $\mathbf{O} \in$ $\{0,1\}^{N \times M}$ be an indicator matrix, where each entry $o_{u, i}$ is an observation indicator: $o_{u, i}=1$ if the rating $r_{u, i}$ is observed, and $o_{u, i}=0$ if the rating $r_{u, i}$ is missing. Given $\mathbf{R}^o$ as the set of the observed entries in the rating matrix $\mathbf{R}$, the rating prediction model aims to train the prediction model that minimizes the prediction inaccuracy $\mathcal{P}$. Nonetheless, as users are free to choose items to rate, leading to the collection of observational data that is always missing not at random (MNAR)~\cite{marlin2009collaborative,de2014reducing, zhang2023mnar}, \emph{e.g.}, the ratings with a lower value are more likely to be missing. Let $\mathcal{O}=\left\{(u, i) \mid (u, i) \in \mathcal{D}, o_{u, i}=1\right\}$ be the set of user-item pairs for the observed ratings, the direct use of the naive estimator $\mathcal{E}_{\mathrm{N}}=\mathcal{E}_{\mathrm{N}}(\hat{\mathbf{R}}, \mathbf{R}^o)=\frac{1}{|\mathcal{O}|} \sum_{(u, i) \in \mathcal{O}} e_{u, i}$ that 
computes on the observed data would yield severely biased estimation~\cite{Schnabel-Swaminathan2016,Wang-Zhang-Sun-Qi2019}.

\subsection{Existing Estimators}
To unbiasedly estimate the prediction inaccuracy $\mathcal{P}$ given the partially observed ratings $\mathbf{R}^o$, the error-imputation-based (EIB) approaches~\cite{Steck2010,Lobato-etal2014} compute an imputed error $\hat e_{u, i}$ for each missing rating, and estimate the prediction inaccuracy with
\begin{align*}
    \mathcal{E}_\mathrm{EIB}(\hat{\mathbf{R}}, \mathbf{R}^o) = \frac{1}{|\cD|} \sum_{(u,i)\in \cD} (o_{u,i} e_{u,i} + (1-o_{u,i}) \hat e_{u,i}),
\end{align*}
which is an unbiased estimator of the prediction inaccuracy when the imputed errors are accurate, \emph{i.e.}, $\hat e_{u, i}=e_{u, i}$.

The inverse-propensity-scoring (IPS) approaches~\cite{Schnabel-Swaminathan2016,saito2020ips} first learn $\hat p_{u, i}$ as the estimate of the propensity $p_{u, i}=\P(o_{u, i}=1\mid x_{u, i})$, \emph{i.e.}, the probability of observing the rating, then inversely weight the prediction error for each observed rating with the learned propensity, and estimate the prediction inaccuracy with
\begin{align*}
    \mathcal{E}_\mathrm{IPS}(\hat{\mathbf{R}}, \mathbf{R}^o) = \frac{1}{|\cD|} \sum_{(u,i)\in \cD} \frac{o_{u,i} e_{u,i}}{\hat p_{u,i}},
    \end{align*}
which is an unbiased estimator of the prediction inaccuracy when the learned propensities are accurate, \emph{i.e.}, $\hat p_{u, i}=p_{u, i}$.

The doubly robust (DR) approaches~\cite{Wang-Zhang-Sun-Qi2019,li2023balancing,li2023propensity, li2023removing} use both the error imputation model and the propensity model to estimate the prediction inaccuracy with
\begin{align*}
\mathcal{E}_\mathrm{DR}(\hat{\mathbf{R}}, \mathbf{R}^o) = \frac{1}{|\cD|} \sum_{(u,i) \in \cD} \left( \hat e_{u,i}  +  \frac{ o_{u,i} (e_{u,i} -  \hat e_{u,i}) }{ \hat p_{u, i} } \right),
\end{align*}
which is an unbiased estimator of the prediction inaccuracy when either the imputed errors or the learned propensities are accurate, \emph{i.e.}, $\hat e_{u, i}=e_{u, i}$ or $\hat p_{u, i}=p_{u, i}$, this is also known as double robustness.

\section{Methodology}
\subsection{Task Formulation under OME}
Despite many methods have been proposed for achieving unbiased learning to tackle the data MNAR problem~\cite{Schnabel-Swaminathan2016,Wang-Zhang-Sun-Qi2019,saito2020doubly,wang2023counterclr,li2024interference}, they ignore an alternative form of bias caused by the inconsistency between the observed ratings and the users' true preferences, also known as noisy feedback or outcome measurement errors (OME). Both data MNAR and OME arise from systematic
bias during the data collection. In RS, two common scenarios that cause incorrect user feedback signals include the influence of public opinions~\cite{zheng2021-distangle}, and the low-quality data collection such as recommender attacks~\cite{gunes2014shilling} or carelessly filling out the after-sales assessments~\cite{wang2021denoising}. Formally, we denote $\mathbf{R}^*$ as the user's true preference matrix, with $r^*_{u, i}$ as its entries, which may deviate from the potentially observed ratings $r_{u, i}$. Let $\P({r}_{u, i}=0 \mid r^{*}_{u, i}=1)=\rho_{01}$ and $ \P({r}_{u, i}=1 \mid r^*_{u, i}=0)=\rho_{10}$ be the false negative rate and false positive rate, respectively, where $\rho_{01}+\rho_{10}<1$, then we have  
\[
\P(r_{u, i}=1\mid x_{u, i})=(1-\rho_{01})\cdot \P(r^*_{u, i}=1\mid x_{u, i})+\rho_{10}\cdot\P(r^*_{u, i}=0\mid x_{u, i}).
\]

To make current debiased recommendation techniques more applicable to real-world scenarios that may contain OME, we then establish a new model evaluation criterion for the unbiased learning under OME, named \emph{true prediction inaccuracy}, formally defined as
\[   \mathcal{P}^*=\mathcal{P}(\hat{\mathbf{R}}, \mathbf{R}^*)  = \frac{1}{|\cD|} \sum_{(u,i)\in \cD} e_{u,i}^*,     \]
where $e_{u,i}^* = \ell(f_{\theta}(x_{u,i}),r^*_{u,i})$ is called the \emph{true prediction error.}

\subsection{Proposed OME-EIB, OME-IPS, and OME-DR Estimators}
To achieve unbiased learning of $r^*_{u, i}$ using MNAR data with OME, a typical class of methods~\cite{natarajan2013learning,liu2015classification} propose the use of a surrogate loss $\tilde{\ell}(f_{\theta}(x_{u, i}), r_{u, i})$ based on the observed labels $r_{u, i}$, which satisfies
\[\mathbb{E}_{r\mid r^*}[\tilde{\ell}(f_{\theta}(x_{u, i}), r_{u, i})]=\ell(f_{\theta}(x_{u, i}), r^*_{u, i}).\]
Considering the cases $r^*_{u, i}=1$ and $r^*_{u, i}=0$ separately, we have
$$
\begin{aligned}
& \left(1-\rho_{01}\right)\cdot \tilde{\ell}(f_{\theta}(x_{u,i}),1)+\rho_{01}\cdot \tilde{\ell}(f_{\theta}(x_{u,i}),0)=\ell(f_{\theta}(x_{u,i}),1), \\
& \left(1-\rho_{10}\right)\cdot \tilde{\ell}(f_{\theta}(x_{u,i}),0)+\rho_{10}\cdot \tilde{\ell}(f_{\theta}(x_{u,i}),1)=\ell(f_{\theta}(x_{u,i}),0).
\end{aligned}
$$
Solving these equations for $\tilde{\ell}(f_{\theta}(x_{u,i}),1)$ and $\tilde{\ell}(f_{\theta}(x_{u,i}),0)$ gives
$$
\begin{aligned}
\tilde{\ell}(f_{\theta}(x_{u,i}),1) & =\frac{\left(1-\rho_{10}\right) \ell(f_{\theta}(x_{u,i}),1)-\rho_{01} \ell(f_{\theta}(x_{u,i}),0)}{1-\rho_{01}-\rho_{10}}, \\
\tilde{\ell}(f_{\theta}(x_{u,i}),0) & =\frac{\left(1-\rho_{01}\right) \ell(f_{\theta}(x_{u,i}),0)-\rho_{10} \ell(f_{\theta}(x_{u,i}),1)}{1-\rho_{01}-\rho_{10}}.
\end{aligned}
$$
However, the existing surrogate loss-based methods require fully observed noisy labels $r_{u,i}$, which prevents the direct use of such methods for RS in the presence of missing data. To fill this gap, motivated by the broad usage of EIB, IPS, and DR estimators in the missing data literature~\cite{bang2005doubly,little2019statistical}, we extend the above surrogate loss-based methods to address the data MNAR and OME in parallel.

Specifically, given knowledge of error parameters $\rho_{01}$ and $\rho_{10}$, let $\tilde e_{u, i}= r_{u,i} \cdot \tilde{\ell}(f_{\theta}(x_{u,i}),1) + (1-r_{u,i}) \cdot \tilde{\ell}(f_{\theta}(x_{u,i}),0)$ be the surrogate loss with the rating $r_{u, i}$, the OME-EIB estimator estimates the true prediction inaccuracy $\mathcal{P}^*=\mathcal{P}(\hat{\mathbf{R}}, \mathbf{R}^*)$ with
\begin{align*}
&\cE_\mathrm{OME-EIB}(\hat{\mathbf{R}}, \mathbf{R}^o; \rho_{01}, \rho_{10}) =\frac{1}{|\cD|} \sum_{(u,i) \in \cD}(1-o_{u,i}) \bar e_{u, i}{}+{}\\
&\frac{1}{|\cD|} \sum_{(u,i) \in \cD}\frac{o_{u, i}r_{u, i}\{\left(1-\rho_{10}\right) \ell(f_{\theta}(x_{u,i}),1)-\rho_{01} \ell(f_{\theta}(x_{u,i}),0)\}}{1-\rho_{01}-\rho_{10}}{}+{}\\
&\frac{1}{|\cD|} \sum_{(u,i) \in \cD}\frac{o_{u, i}(1-r_{u, i})\{\left(1-\rho_{01}\right) \ell(f_{\theta}(x_{u,i}),0)-\rho_{10} \ell(f_{\theta}(x_{u,i}),1)\}}{1-\rho_{01}-\rho_{10}},  
\end{align*}
where $\bar e_{u, i}$ is the imputed error for estimating $\tilde e_{u, i}$, and OME-EIB is an unbiased estimator of the true prediction inaccuracy when the imputed errors are accurate, \emph{i.e.}, $\bar e_{u, i}=\tilde e_{u, i}$. Since the proof is not trivial, we postpone to show unbiasedness of OME-EIB estimator and the following OME-IPS and OME-DR estimators in Sec.~\ref{sec:3.4}.

Similarly, the OME-IPS estimator estimates $\mathcal{P}^*=\mathcal{P}(\hat{\mathbf{R}}, \mathbf{R}^*)$ with
{\begin{align*}
&\cE_\mathrm{OME-IPS}(\hat{\mathbf{R}}, \mathbf{R}^o; \rho_{01}, \rho_{10}){}={}\\
&\frac{1}{|\cD|} \sum_{(u,i) \in \cD}\frac{o_{u, i}r_{u, i}}{\hat p_{u, i}}\cdot \frac{\left(1-\rho_{10}\right) \ell(f_{\theta}(x_{u,i}),1)-\rho_{01} \ell(f_{\theta}(x_{u,i}),0)}{1-\rho_{01}-\rho_{10}}{}+{}\\
&\frac{1}{|\cD|} \sum_{(u,i) \in \cD}\frac{o_{u, i}(1-r_{u, i})}{\hat p_{u, i}}\cdot\frac{\left(1-\rho_{01}\right) \ell(f_{\theta}(x_{u,i}),0)-\rho_{10} \ell(f_{\theta}(x_{u,i}),1)}{1-\rho_{01}-\rho_{10}},\end{align*}}where $\hat p_{u, i}$ is the learned propensity for estimating $p_{u, i}$, and OME-IPS is an unbiased estimator of the true prediction inaccuracy when the learned propensities are accurate, \emph{i.e.}, $\hat p_{u, i}=p_{u, i}$.

The OME-DR estimator estimates $\mathcal{P}^*=\mathcal{P}(\hat{\mathbf{R}}, \mathbf{R}^*)$ with

\begin{align*}
&\cE_\mathrm{OME-DR}(\hat{\mathbf{R}}, \mathbf{R}^o; \rho_{01}, \rho_{10})=\frac{1}{|\cD|} \sum_{(u,i) \in \cD} \left(1-\frac{o_{u, i}}{\hat p_{u, i}}\right)\bar e_{u, i}{}+{}\\
&\frac{1}{|\cD|} \sum_{(u,i) \in \cD}\frac{o_{u, i}r_{u, i}}{\hat p_{u, i}}\cdot \frac{\left(1-\rho_{10}\right) \ell(f_{\theta}(x_{u,i}),1)-\rho_{01} \ell(f_{\theta}(x_{u,i}),0)}{1-\rho_{01}-\rho_{10}}{}+{}\\
&\frac{1}{|\cD|} \sum_{(u,i) \in \cD}\frac{o_{u, i}(1-r_{u, i})}{\hat p_{u, i}}\cdot\frac{\left(1-\rho_{01}\right) \ell(f_{\theta}(x_{u,i}),0)-\rho_{10} \ell(f_{\theta}(x_{u,i}),1)}{1-\rho_{01}-\rho_{10}},
\end{align*}
which is an unbiased estimate of the true prediction inaccuracy when either the imputed errors or the learned propensities are accurate, \emph{i.e.}, $\bar e_{u, i}=\tilde e_{u, i}$ or $\hat p_{u, i}=p_{u, i}$.
    
\subsection{Identification and Estimation of $\rho_{01}$ and $\rho_{10}$}\label{sec:3.3}
The proposed OME-EIB, OME-IPS, and OME-DR estimators require knowledge of the known error parameters $\rho_{01}$ and $\rho_{10}$, which are usually not directly available from the collected data. By building upon the existing weak separability assumption~\cite{menon2015learning}, we present a data-driven identification and estimation method of $\rho_{01}$ and $\rho_{10}$.

We impose the following weak separability assumption that
\[\inf _{(u, i) \in \cD} \P(r^*_{u, i}=1\mid x_{u, i})=0 \quad \text{and}\quad  \sup _{(u, i) \in \cD} \P(r^*_{u, i}=1\mid x_{u, i})=1,\]
which also known as mutual irreducibility~\cite{scott2013classification,scott2015rate} in the observational label noise literature. This does not require the true ratings to be separable, \emph{i.e.}, $\P(r^*_{u, i}=1\mid x_{u, i})\in \{0,1\}$ for all user-item pairs, but instead stipulates that there exist "perfectly positive" and "perfectly negative" samples among the entire user-item pairs. In real-world recommendation scenarios, the weak separability assumption is easily satisfied, providing there exists at least one "perfectly positive" feedback and one "perfectly negative" feedback among the thousands of collected ratings. For example, in the movie rating scenario, "perfectly positive" feedback refers to at least one of the users who made a positive review of the movie is fully reliable, and we do not need to know who that user exactly is.

By noting the linkage between the observed and the true ratings
\begin{align*}
\P(r_{u, i}=1\mid x_{u, i})={}&{}(1-\rho_{01})\P(r^*_{u, i}=1\mid x_{u, i})+\rho_{10}\P(r^*_{u, i}=0\mid x_{u, i})\\
={}&{}(1-\rho_{01}-\rho_{10})\P(r^*_{u, i}=1\mid x_{u, i})+\rho_{10},
\end{align*}
which demonstrates the monotonicity between $\P(r_{u, i}=1\mid x_{u, i})$ and $\P(r^*_{u, i}=1\mid x_{u, i})$ under $\rho_{01}+\rho_{10}<1$, we have
\begin{align*}
(u^{<}, i^{<})&=\arg\min_{(u, i) \in \cD} \P(r_{u, i}=1\mid x_{u, i})=\arg\min_{(u, i) \in \cD} \P(r^*_{u, i}=1\mid x_{u, i}),\\
(u^{>}, i^{>})& =\arg\max_{(u, i) \in \cD} \P(r_{u, i}=1\mid x_{u, i})=\arg\max_{(u, i) \in \cD} \P(r^*_{u, i}=1\mid x_{u, i}).
\end{align*}
Then, we can identify the error parameters $\rho_{01}$ and $\rho_{10}$ via
\[
\rho_{01}=1-{}\P(r_{u^{>}, i^{>}}=1\mid x_{u^{>}, i^{>}}) \quad \text {and}\quad \rho_{10}=\P(r_{u^{<}, i^{<}}=1\mid x_{u^{<}, i^{<}}),
\]
where $\P(r_{u^{>}, i^{>}}=1\mid x_{u^{>}, i^{>}})$ and $\P(r_{u^{<}, i^{<}}=1\mid x_{u^{<}, i^{<}})$ can be unbiasedly estimated from the existing EIB, IPS, or DR estimators for estimating $\P(r_{u, i}=1\mid x_{u, i})$ without considering the OME.

\subsection{Theoretical Analyses}\label{sec:3.4}
Since the OME-DR estimator degenerates to the OME-IPS estimator when $\bar e_{u, i}=0$, and degenerates to the OME-EIB estimator when $\hat p_{u, i}=1$, without loss of generality, we only analyze the explicit bias form of the OME-DR estimator. Following existing literature~\cite{Schnabel-Swaminathan2016,Wang-Zhang-Sun-Qi2019}, we assume that the indicator matrix $\mathbf{O}$ contains independent random variables and each $o_{u, i}$ follows a Bernoulli distribution with probability $p_{u, i}$. In addition, due to the presence of OME, we also consider the randomness of the potentially observed ratings $r_{u, i}$ given the true ratings $r^*_{u, i}$, \emph{e.g.}, $r_{u, i}=1$ with probability $1-\rho_{01}$ given $r^*_{u, i}=1$, and $r_{u, i}=0$ with probability $\rho_{01}$ given $r^*_{u, i}=1$.

\begin{theorem}[Bias of OME-DR Estimator]\label{thm3.1} Given $\hat \rho_{01}$ and  $\hat \rho_{10}$ with $\hat \rho_{01}+\hat \rho_{10}<1$, imputed errors $\bar{\mathbf{E}}$ and learned propensities $\hat{\mathbf{P}}$ with $\hat{p}_{u, i}>0$ for all user-item pairs, the bias of the OME-DR estimator is
\begin{align*}
&\operatorname{Bias}[\mathcal{E}_{\mathrm{OME-DR}}(\hat{\mathbf{R}}, \mathbf{R}^o; \hat \rho_{01}, \hat \rho_{10})]=
\frac{1}{|\cD|}\left|\sum_{(u,i) \in \cD}\left(1-\frac{o_{u, i}}{\hat p_{u, i}}\right)\bar e_{u, i}\right.+\\
&\sum_{(u,i): r^*_{u,i}=1}\left(\frac{p_{u, i}\omega_{11}-\hat p_{u, i}}{\hat p_{u, i}}\ell(f_{\theta}(x_{u,i}),1)+\frac{p_{u, i}\omega_{01}}{\hat p_{u, i}}\ell(f_{\theta}(x_{u,i}),0)\right)+\\
&\left.\sum_{(u,i): r^*_{u,i}=0}\left(\frac{p_{u, i}\omega_{10}}{\hat p_{u, i}}\ell(f_{\theta}(x_{u,i}),1)+\frac{p_{u, i}\omega_{00}-\hat p_{u, i}}{\hat p_{u, i}}\ell(f_{\theta}(x_{u,i}),0)\right)\right|,
\end{align*}
where $\omega_{11}$, $\omega_{01}$, $\omega_{10}$, and $\omega_{00}$ are given by
\begin{align*}
\omega_{11}&=\frac{1-\rho_{01}-\hat \rho_{10}}{1-\hat \rho_{01}-\hat \rho_{10}}, \quad \omega_{01}= \frac{\rho_{01}-\hat \rho_{01}}{1-\hat \rho_{01}-\hat \rho_{10}},\\
\omega_{10}&=\frac{\rho_{10}-\hat \rho_{10}}{1-\hat \rho_{01}-\hat \rho_{10}}, \quad \omega_{00}=\frac{1-\hat \rho_{01}- \rho_{10}}{1-\hat \rho_{01}-\hat \rho_{10}}.
\end{align*}
\end{theorem}
The bias of the OME-DR estimator includes the three terms: (1) The first term shares a similar form to the bias of the previous IPS estimator and leads to smaller bias when $\hat p_{u, i}\approx p_{u, i}$. (2) The second term is novel for OME, specifically focusing on the estimated false negative rate $\hat \rho_{01}$ that corresponds to the positive samples $r^*_{u, i}=1$. Moreover, we find that $\omega_{11}=1$ and $\omega_{01}=0$ when the the estimated false negative rate $\hat \rho_{01}$ is accurate, \emph{i.e.}, $\hat \rho_{01}=\rho_{01}$, which results in smaller bias when $\bar e_{u, i}\approx\tilde e_{u, i}$. (3) The third term is similar to the second term, but instead focuses on the estimated false positive rate $\hat \rho_{10}$ that corresponds to the negative samples $r^*_{u, i}=0$, and also results in smaller bias when $\bar e_{u, i}\approx\tilde e_{u, i}$. Given the importance of bias derivation for constructing estimators under OME, we provide a proof sketch as below (see Appendix \ref{app-a} for more details).

\begin{proof}[Proof Sketch]
By definition, bias of the OME-DR estimator is
\begin{align*}
\operatorname{Bias}[\mathcal{E}_{\mathrm{OME-DR}}(\hat{\mathbf{R}}, \mathbf{R}^o; \hat \rho_{01}, \hat \rho_{10})]=\left|\mathbb{E}_{\mathbf{R},\mathbf{O}}\left[\mathcal{E}_{\mathrm{OME-DR}}\right]-\mathcal{P}^*\right|.
\end{align*}
Then, we can derive the bias of the OME-DR estimator as follows
\begin{align*}
&\operatorname{Bias}[\mathcal{E}_{\mathrm{OME-DR}}(\hat{\mathbf{R}}, \mathbf{R}^o; \hat \rho_{01}, \hat \rho_{10})]=\left|\mathbb{E}_{\mathbf{R}\mid \mathbf{O}}\left[\mathbb{E}_{\mathbf{O}}\left[\mathcal{E}_{\mathrm{OME-DR}}\right]\right]-\mathcal{P}^*\right|\\
={}&{}\left|\frac{1}{|\cD|} \sum_{(u,i): r^*_{u,i}=1} \mathbb{E}_{\mathbf{R}\mid \mathbf{O}}\left[\left(1-\frac{p_{u, i}}{\hat p_{u, i}}\right)\bar e_{u, i}+\frac{p_{u, i}\tilde e_{u, i}}{\hat p_{u, i}}\right]-\ell(f_{\theta}(x_{u,i}),1)\right.\\
+{}&{}\left.\frac{1}{|\cD|}\sum_{(u,i): r^*_{u,i}=0}\mathbb{E}_{\mathbf{R}\mid \mathbf{O}}\left[\left(1-\frac{p_{u, i}}{\hat p_{u, i}}\right)\bar e_{u, i}+\frac{p_{u, i}\tilde e_{u, i}}{\hat p_{u, i}}\right]-\ell(f_{\theta}(x_{u,i}),0)\right|.
\end{align*}
For the user-item pairs with $r^*_{u,i}=1$, we have
\begin{align*}
&\mathbb{E}_{\mathbf{R}\mid \mathbf{O}}\left[\tilde e_{u,i}\right]\\
={}&{}\left(1-\rho_{01}\right) \cdot\tilde{\ell}(f_{\theta}(x_{u,i}),1; \hat \rho_{01}, \hat \rho_{10})+\rho_{01}\cdot \tilde{\ell}(f_{\theta}(x_{u,i}),0; \hat \rho_{01}, \hat \rho_{10})\\
={}&{}\left(1-\rho_{01}\right) \cdot \frac{\left(1-\hat \rho_{10}\right) \ell(f_{\theta}(x_{u,i}),1)-\hat \rho_{01} \ell(f_{\theta}(x_{u,i}),0)}{1-\hat \rho_{01}-\hat \rho_{10}}\\
&{}+\rho_{01}\cdot\frac{\left(1-\hat \rho_{01}\right) \ell(f_{\theta}(x_{u,i}),0)-\hat \rho_{10} \ell(f_{\theta}(x_{u,i}),1)}{1-\hat \rho_{01}-\hat \rho_{10}}\\
={}&{}\frac{1-\rho_{01}-\hat \rho_{10}}{1-\hat \rho_{01}-\hat \rho_{10}}\ell(f_{\theta}(x_{u,i}),1)+\frac{\rho_{01}-\hat \rho_{01}}{1-\hat \rho_{01}-\hat \rho_{10}}\ell(f_{\theta}(x_{u,i}),0)\\
:={}&{}\omega_{11}\ell(f_{\theta}(x_{u,i}),1)+\omega_{01}\ell(f_{\theta}(x_{u,i}),0),
\end{align*}
where $\omega_{11}$ and $\omega_{01}$ are given by
\begin{align*}
\omega_{11}=\frac{1-\rho_{01}-\hat \rho_{10}}{1-\hat \rho_{01}-\hat \rho_{10}}, \quad \omega_{01}= \frac{\rho_{01}-\hat \rho_{01}}{1-\hat \rho_{01}-\hat \rho_{10}}.
\end{align*}
Similar results hold for the user-item pairs with $r^*_{u,i}=0$ that
\begin{align*}
&\mathbb{E}_{\mathbf{R}\mid \mathbf{O}}\left[\tilde e_{u,i}\right]\\
={}&{}\frac{\rho_{10}-\hat \rho_{10}}{1-\hat \rho_{01}-\hat \rho_{10}}\ell(f_{\theta}(x_{u,i}),1)+\frac{1-\hat \rho_{01}- \rho_{10}}{1-\hat \rho_{01}-\hat \rho_{10}}\ell(f_{\theta}(x_{u,i}),0)\\
:={}&{}\omega_{10}\ell(f_{\theta}(x_{u,i}),1)+\omega_{00}\ell(f_{\theta}(x_{u,i}),0),
\end{align*}
where $\omega_{10}$ and $\omega_{00}$ are given by
\begin{align*}
\omega_{10}=\frac{\rho_{10}-\hat \rho_{10}}{1-\hat \rho_{01}-\hat \rho_{10}}, \quad \omega_{00}=\frac{1-\hat \rho_{01}- \rho_{10}}{1-\hat \rho_{01}-\hat \rho_{10}}.
\end{align*}
This completes the proof.    
\end{proof}
We formally describe double robustness under OME as follows.
\begin{corollary}[Double Robustness] Given $\hat \rho_{01}=\rho_{01}$ and $\hat \rho_{10}=\rho_{10}$, the OME-DR estimator is unbiased when either imputed errors $\bar{\mathbf{E}}$ or learned propensities $\hat{\mathbf{P}}$
are accurate for all user-item pairs.
\end{corollary}
The above result is obtained via substituting either $\bar e_{u, i}=\tilde e_{u, i}$ or $\hat p_{u, i}=p_{u, i}$ into the bias of the OME-DR estimator in Theorem \ref{thm3.1}.
Given the estimated error parameters $\hat \rho_{01}$ and $\hat \rho_{01}$, we obtain the optimal prediction model under OME by minimizing the OME-DR estimator over a hypothesis space $\mathcal{F}$ of the prediction models $f_{\theta}$
\begin{align*}
\hat{\mathbf{R}}^{\ddagger}=\arg\min_{f_\theta\in\mathcal{F}}\{\mathcal{E}_{\mathrm{OME-DR}}(\hat{\mathbf{R}}, \mathbf{R}^o; \hat \rho_{01}, \hat \rho_{10})\}.    
\end{align*}
We next derive the generalization bound of the optimal prediction model in terms of the empirical Rademacher complexity~\cite{shalev2014understanding}. The basic idea of proving the performance guarantee under OME is to exploit the inheritance of the Lipschitz continuity from the ture prediction loss $\ell(f_{\theta}(x_{u,i}), r^*_{u, i})$ to the surrogate loss $\tilde{\ell}(f_{\theta}(x_{u,i}),r_{u, i})$.

\begin{lemma}[Lipschitz Continuity] Given $\hat \rho_{01}$ and  $\hat \rho_{10}$ with $\hat \rho_{01}+\hat \rho_{10}<1$, if $\ell(f_{\theta}(x_{u,i}), r^*_{u, i})$ is $L$-Lipschitz in $f_{\theta}(x_{u,i})$ for all $r^*_{u, i}$, then $\tilde{\ell}(f_{\theta}(x_{u,i}),r_{u, i})$ is $\frac{2L}{1-\hat \rho_{01}-\hat \rho_{10}}$-Lipschitz in $f_{\theta}(x_{u,i})$ for all $r_{u, i}$.
\end{lemma}

The above lemma immediately leads to a generalization bound with respect to the true ratings by using the contraction principle for Rademacher complexity (see Appendix \ref{app-a} for proofs).

\begin{theorem}[Generalization Bound] Given $\hat \rho_{01}$ and  $\hat \rho_{10}$ with $\hat \rho_{01}+\hat \rho_{10}<1$, suppose $\ell(f_{\theta}(x_{u,i}), r^*_{u, i})$ is $L$-Lipschitz in $f_{\theta}(x_{u,i})$ for all $r^*_{u, i}$, and $\hat p_{u, i}\geq C_p$, $|\tilde{\ell}(f_{\theta}(x_{u,i}),r_{u, i})|\leq C_l$ for all $r_{u, i}$, then with probability $1-\eta$, the true prediction inaccuracy $\mathcal{P}(\hat{\mathbf{R}}^\ddagger, \mathbf{R}^*)$ of the optimal prediction matrix using the OME-DR estimator with imputed errors $\bar{\mathbf{E}}$ and learned propensities $\hat{\mathbf{P}}$ has the upper bound
\begin{align*}
&\mathcal{E}_{\mathrm{OME-DR}}(\hat{\mathbf{R}}^\ddagger, \mathbf{R}^o; \hat \rho_{01}, \hat \rho_{10})+\operatorname{Bias}[\mathcal{E}_{\mathrm{OME-DR}}(\hat{\mathbf{R}}^\ddagger, \mathbf{R}^o; \hat \rho_{01}, \hat \rho_{10})]+\\
&\left(1+\frac{2}{C_p}\right)\left[\frac{4L}{1-\hat \rho_{01}-\hat \rho_{10}}\mathcal{R}(\mathcal{F})+\left(C_l+\frac{4L}{1-\hat \rho_{01}-\hat \rho_{10}}\right)\sqrt{\frac{2\log (4 / \eta)}{|\mathcal{D}|}}\right],
\end{align*}
where $\mathcal{R}(\mathcal{F})$ is the empirical Rademacher complexity defined as
    \[\mathcal{R}(\mathcal{F})=\mathbb{E}_{\mathbf{\sigma} \sim\{-1,+1\}^{|\mathcal{D}|}} \sup _{f_\theta \in \mathcal{F}}\left[\frac{1}{|\mathcal{D}|} \sum_{(u, i) \in \mathcal{D}} \sigma_{u, i} f_{\theta}(x_{u, i})\right],\]
in which $\mathbf{\sigma}=\{\sigma_{u, i}: (u, i)\in \cD\}$, and $\sigma_{u, i}$ are independent uniform random variables taking values in $\{-1,+1\}$. The random variables $\sigma_{u, i}$ are called Rademacher variables.
\end{theorem}
\begin{table*}[]
\vspace{-6pt}
\centering
\caption{RE on \textsc{ML-100K} dataset with $\rho_{01}$ = 0.2 and $\rho_{10}$ = 0.1. The best two results are bolded and the best baseline is underlined.}
\vspace{-10pt}
\setlength{\tabcolsep}{7pt}
\resizebox{0.968\linewidth}{!}{
\begin{tabular}{l|cccccc}
\toprule
  & ROTATE  & SKEW   & CRS  & ONE & THREE  & FIVE 
\\ \midrule
Naive \cite{koren2009matrix} & 0.125 $\pm$ 0.002 & 0.179 $\pm$ 0.001 & 0.175 $\pm$ 0.002 & 0.241 $\pm$ 0.002 & 0.264 $\pm$ 0.003 & 0.299 $\pm$ 0.003  \\

OME ($\hat{\rho}_{01}, \hat{\rho}_{10}$) & 0.087 $\pm$ 0.006 & 0.163 $\pm$ 0.002 & 0.104 $\pm$ 0.014 & 0.161 $\pm$ 0.013 & 0.179 $\pm$ 0.014 & 0.213 $\pm$ 0.014  \\

OME ($\rho_{01}, {\rho}_{10}$) & \underline{0.024 $\pm$ 0.003} & 0.136 $\pm$ 0.002 & 0.105 $\pm$ 0.004 & \underline{0.067 $\pm$ 0.004} & \underline{0.076 $\pm$ 0.004} & \underline{0.100 $\pm$ 0.004}  \\

MRDR \cite{MRDR} &  0.098 $\pm$ 0.003 & 0.089 $\pm$ 0.001 & \underline{0.099 $\pm$ 0.002} & 0.172 $\pm$ 0.004 & 0.190 $\pm$ 0.004 & 0.194 $\pm$ 0.004 \\
SDR \cite{SDR} & 0.097 $\pm$ 0.002 & 0.089 $\pm$ 0.001 & 0.102 $\pm$ 0.002 & 0.172 $\pm$ 0.003 & 0.191 $\pm$ 0.003 & 0.194 $\pm$ 0.003 \\ 
TDR \cite{TDR} & 0.092 $\pm$ 0.002 & \underline{0.080 $\pm$ 0.001} & 0.103 $\pm$ 0.002 & 0.171 $\pm$ 0.003 & 0.189 $\pm$ 0.003 & 0.195 $\pm$ 0.003 \\
\midrule
EIB \cite{Steck2010} & 0.366 $\pm$ 0.001 & 0.256 $\pm$ 0.001 & 0.150 $\pm$ 0.001 & 0.562 $\pm$ 0.001 & 0.605 $\pm$ 0.001 & 0.635 $\pm$ 0.001 \\
OME-EIB ($\hat{\rho}_{01}, \hat{\rho}_{10}$) & 0.362 $\pm$ 0.001 & 0.255 $\pm$ 0.001 & 0.144 $\pm$ 0.001 & 0.554 $\pm$ 0.001 & 0.598 $\pm$ 0.001 & 0.627 $\pm$ 0.001 \\
OME-EIB (${\rho}_{01}, {\rho}_{10}$) & 0.357 $\pm$ 0.001 & 0.253 $\pm$ 0.001 & 0.144 $\pm$ 0.001 & 0.546 $\pm$ 0.001 & 0.589 $\pm$ 0.001 & 0.618 $\pm$ 0.001 \\
\midrule
IPS  \cite{Schnabel-Swaminathan2016} & 0.110 $\pm$ 0.002 & 0.116 $\pm$ 0.002 & 0.134 $\pm$ 0.003 & 0.212 $\pm$ 0.004 & 0.231 $\pm$ 0.004 & 0.254 $\pm$ 0.004 \\
OME-IPS ($\hat{\rho}_{01}, \hat{\rho}_{10}$) & 0.060 $\pm$ 0.003 & 0.096 $\pm$ 0.002 & 0.075 $\pm$ 0.017 & 0.111 $\pm$ 0.006 & 0.125 $\pm$ 0.007 & 0.144 $\pm$ 0.007 \\ 
OME-IPS (${\rho}_{01}, {\rho}_{10}$) & \textbf{0.013 $\pm$ 0.003$^*$} & \textbf{0.068 $\pm$ 0.003$^*$} & \textbf{0.052 $\pm$ 0.004$^*$} & \textbf{0.034 $\pm$ 0.005$^*$} & \textbf{0.038 $\pm$ 0.006$^*$} & \textbf{0.050 $\pm$ 0.005$^*$} \\ 
\midrule
DR \cite{saito2020doubly} & 0.106 $\pm$ 0.002 & 0.087 $\pm$ 0.001 & 0.104 $\pm$ 0.002 & 0.190 $\pm$ 0.003 & 0.209 $\pm$ 0.003 & 0.216 $\pm$ 0.003 \\
OME-DR ($\hat{\rho}_{01}, \hat{\rho}_{10}$) & 0.056 $\pm$ 0.004 & 0.067 $\pm$ 0.001 & 0.045 $\pm$ 0.018 & 0.090 $\pm$ 0.006 & 0.102 $\pm$ 0.006 & 0.106 $\pm$ 0.007 \\ 
OME-DR (${\rho}_{01}, {\rho}_{10}$) & \textbf{0.009 $\pm$ 0.003$^*$} & \textbf{0.039 $\pm$ 0.002$^*$} & \textbf{0.022 $\pm$ 0.003$^*$} & \textbf{0.013 $\pm$ 0.004$^*$} & \textbf{0.016 $\pm$ 0.005$^*$} & \textbf{0.012 $\pm$ 0.004$^*$} \\ 
 \bottomrule 
 \end{tabular}}
 \label{tab:semi1}%
\begin{tablenotes}
\scriptsize
\item Note: * means statistically significant results ($\text{p-value} \leq 0.05$) using the paired-t-test compared with the best baseline.
\end{tablenotes}
\vspace{-6pt}
\end{table*}
\subsection{Alternating Denoise Training Approach}
\begin{algorithm}[t]
\caption{Alternating Denoise Training with OME-DR.}
\label{alg.1}
\LinesNumbered 
\KwIn{observed ratings $\mathbf{R}^o$, learned propensities $\hat p_{u,i}$, initial estimates $\hat \rho_{01}$ and $\hat \rho_{10}$, a pre-trained model $h_{\theta^\prime}(x_{u, i})$ for estimating $\P(r_{u, i}=1\mid x_{u, i})$ via existing method.}
\While{stopping criteria is not satisfied}{
    \For{number of steps for training the denoising prediction model}{Sample a batch $\{(u_{j}, i_{j})\}_{j=1}^{J}$ from $\mathcal{D}$\;
    Update $\theta$ by descending along the gradient $\nabla_{\theta} \mathcal{E}_\mathrm{OME-DR}(\theta, \phi;\hat \rho_{01}, \hat \rho_{10})$\;
    $(u^{<}, i^{<})\gets\arg\min_{(u, i) \in \cD} f_{\theta}(x_{u, i})$\; 
    $(u^{>}, i^{>})\gets\arg\max_{(u, i) \in \cD} f_{\theta}(x_{u, i})$\;} 

\For{number of steps for training the denoising imputation model}{
        $\hat \rho_{01}\gets 1-h_{\theta^\prime}(x_{u^{>}, i^{>}})$ and $\hat \rho_{10}\gets h_{\theta^\prime}(x_{u^{<}, i^{<}})$\;
    Sample a batch $\{(u_{k}, i_{k})\}_{k=1}^{K}$ from {$\mathcal{O}$}\;
    Update $\phi$ by descending along the gradient $\nabla_{\phi} \cL_{\bar e}(\theta, \phi;\hat \rho_{01}, \hat \rho_{10})$\; 
    }
}
\end{algorithm}
We further propose an alternating denoise training approach to achieve unbiased learning of the prediction model, which corrects for data MNAR and OME in parallel. Specifically, we first train a propensity model using the observed MNAR data. Based on the estimated propensities, the denoising prediction and imputation models are alternately updated, which also facilitates the accurate estimations $\hat \rho_{01}$ and $\hat \rho_{10}$ of the error parameters $\rho_{01}$ and $\rho_{10}$.

\vspace{4pt}\noindent{\textbf{Propensity Estimation via Logistic Regression.}} Since unbiased ratings are difficult to be collected in the real-world scenarios~\cite{saito2019towards}, following previous studies~\cite{Schnabel-Swaminathan2016,TDR}, we adopt logistic regression to train a propensity model $\hat p_{u,i}=\sigma(w^\top x_{u, i}+\beta_u+\gamma_i)$ parameterized by $\psi=(w, \beta_1, \dots, \beta_N, \gamma_1, \dots, \gamma_M)$ using the observed MNAR data, where $\sigma(\cdot)$ is the sigmoid function, $w\in \mathbb{R}^{N+M}$ is the weight vector, $\beta_u$ and $\gamma_i$ are the user-specific and item-specific constants, respectively. We train the propensity model by
minimizing the loss 
\begin{align*}
\mathcal{L}_p(\psi)=-\frac{1}{|\cD|}\sum_{(u, i)\in\cD} o_{u, i}\log (\hat p_{u, i})+(1-o_{u, i})\log (1-\hat p_{u, i}).
\end{align*}

\noindent{\textbf{Denoising Prediction Model Training.}}  Based on the learned propensities $\hat p_{u, i}$, imputed errors $\bar e_{u, i}(\phi)$, and initial estimates $\hat \rho_{01}$ and $\hat \rho_{10}$, we train the denoising prediction model by minimizing the estimated true prediction inaccuracy from the proposed OME-DR estimator $\mathcal{E}_\mathrm{OME-DR}(\theta, \phi;\hat \rho_{01}, \hat \rho_{10})$. To obtain more accurate estimates $\hat \rho_{01}$ and $\hat \rho_{10}$ from the training loop, we then compute the minimum and maximum predicted ratings in the training batch to update $(u^{<}, i^{<})$ and $(u^{>}, i^{>})$, respectively. 
The choice not to employ a separate noisy rating prediction model for computing $(u^{<}, i^{<})$ and $(u^{>}, i^{>})$ is guaranteed by the 
monotonicity between $\P(r_{u, i}=1\mid x_{u, i})$ and $\P(r^*_{u, i}=1\mid x_{u, i})$ (see Sec. \ref{sec:3.3} for proofs), where the latter has a larger gap between the minimum and the maximum, making it easier to distinguish $(u^{<}, i^{<})$ and $(u^{>}, i^{>})$.

\vspace{4pt}\noindent{\textbf{Denoising Imputation Model Training.}} Given a pre-trained model $h_{\theta^\prime}(x_{u, i})$ for estimating $\P(r_{u, i}=1\mid x_{u, i})$ via existing methods~\cite{Schnabel-Swaminathan2016,Wang-Zhang-Sun-Qi2019}, we update the estimates $\hat \rho_{01}$ and $\hat \rho_{10}$ by computing $1-h_{\theta^\prime}(x_{u^{>}, i^{>}})$ and $h_{\theta^\prime}(x_{u^{<}, i^{<}})$, respectively. We finally train the denoising imputation model $\bar e_{u, i}(\phi)$ by minimizing the loss 
\[
\cL_{\bar e}(\theta, \phi;\hat \rho_{01}, \hat \rho_{10})=\frac{1}{|\cD|} \sum_{(u,i)\in \cD} \frac{o_{u, i}(\tilde e_{u, i}-{\bar{e}}_{u, i})^2}{\hat p_{u, i}}.
\]
The overall training process with OME-DR is summarized in Alg. \ref{alg.1}.

\begin{figure*}[t]
\centering
\vspace{-12pt}
\resizebox{1\linewidth}{!}{
\subfigure[Estimation error of $\hat{\rho}_{01}$ and $\hat{\rho}_{10}$]{
\begin{minipage}[t]{0.2275\linewidth}
\centering
\includegraphics[width=1\textwidth]{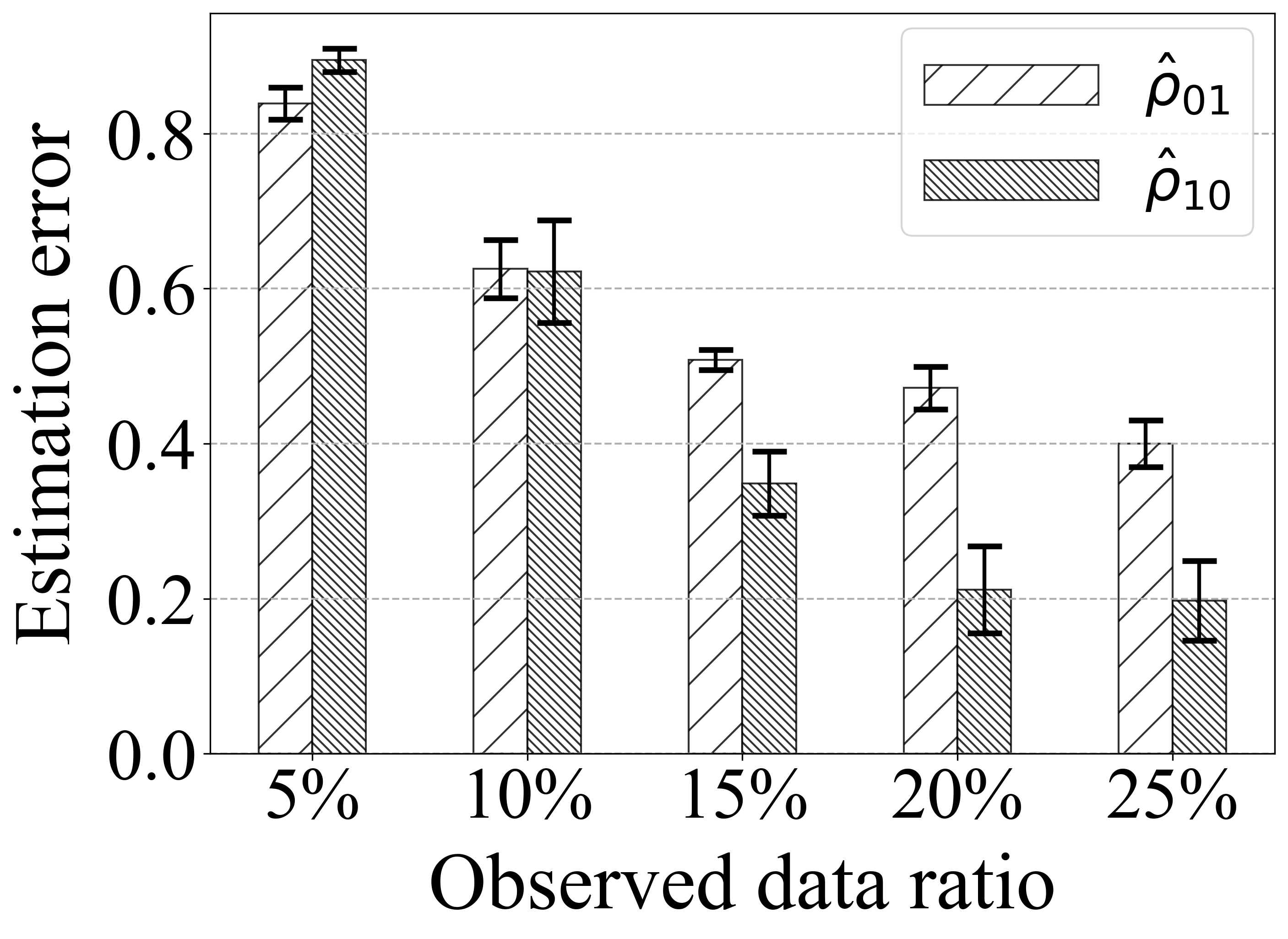}
\end{minipage}%
}%
\subfigure[RE on EIB-based methods]{
\begin{minipage}[t]{0.225\linewidth}
\centering
\includegraphics[width=1\textwidth]{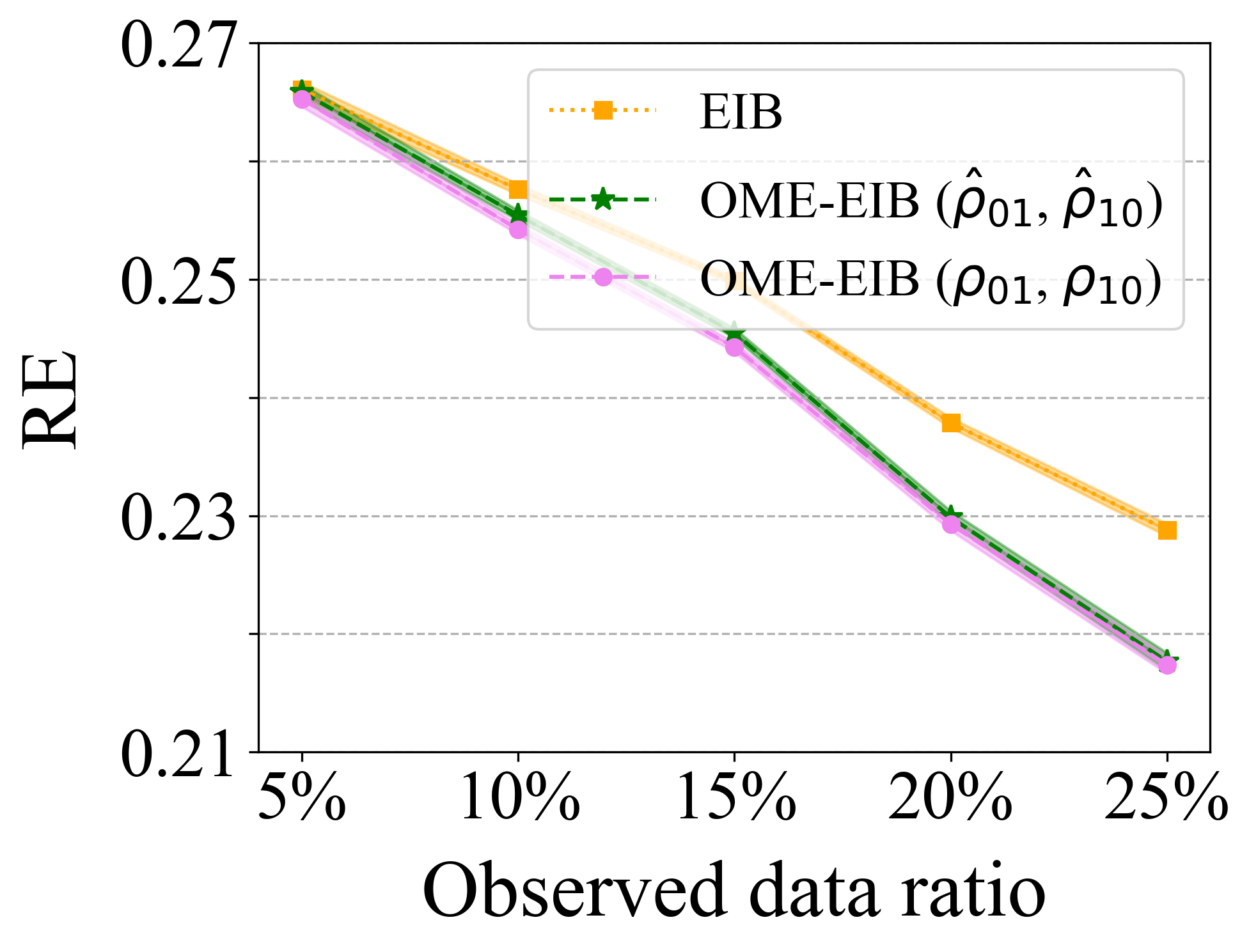}
\end{minipage}%
}
\centering
\subfigure[RE on IPS-based methods]{
\begin{minipage}[t]{0.23\linewidth}
\centering
\includegraphics[width=1\textwidth]{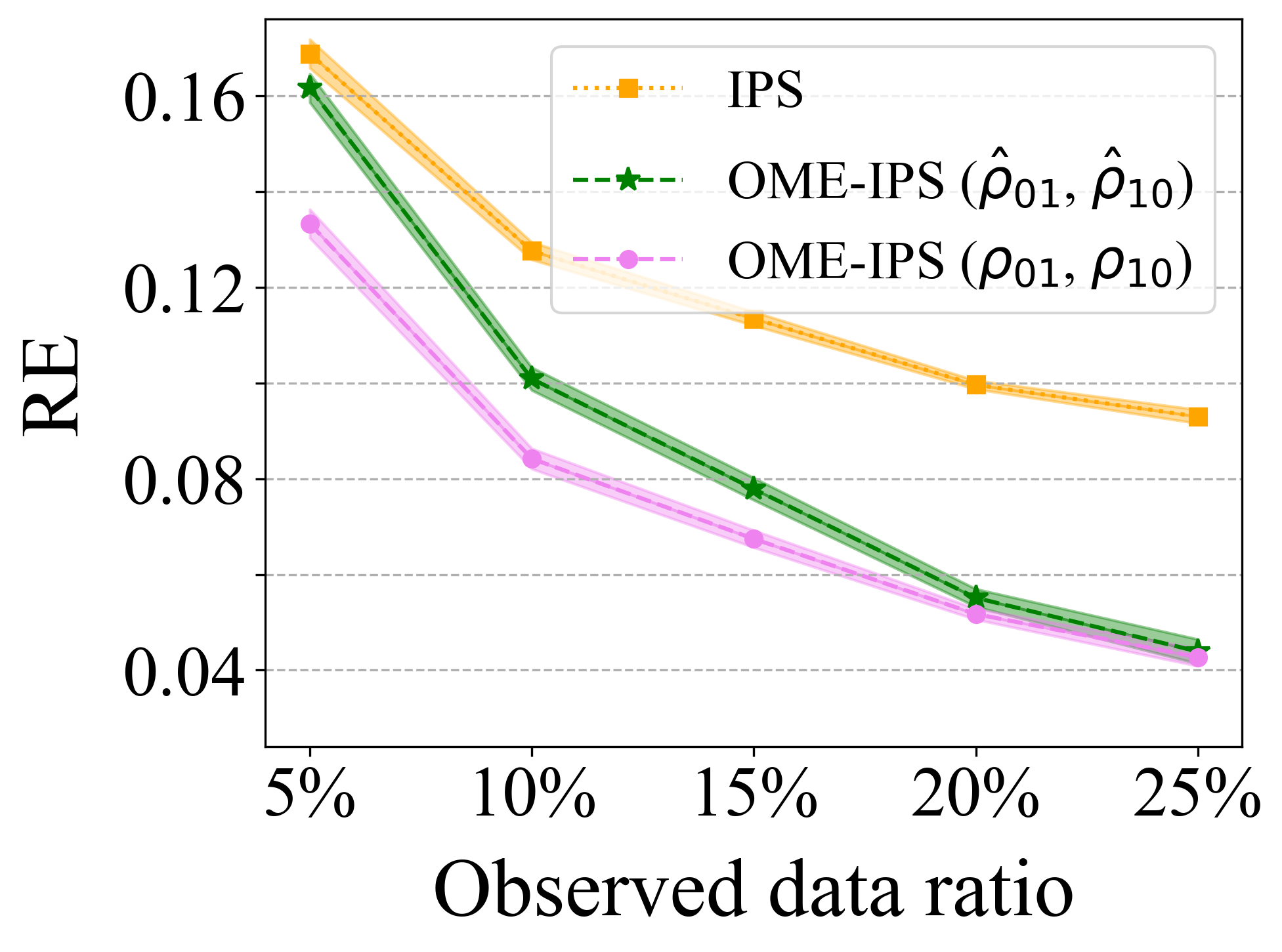}
\end{minipage}%
}%
\subfigure[RE on DR-based methods]{
\begin{minipage}[t]{0.225\linewidth}
\centering
\includegraphics[width=1\textwidth]{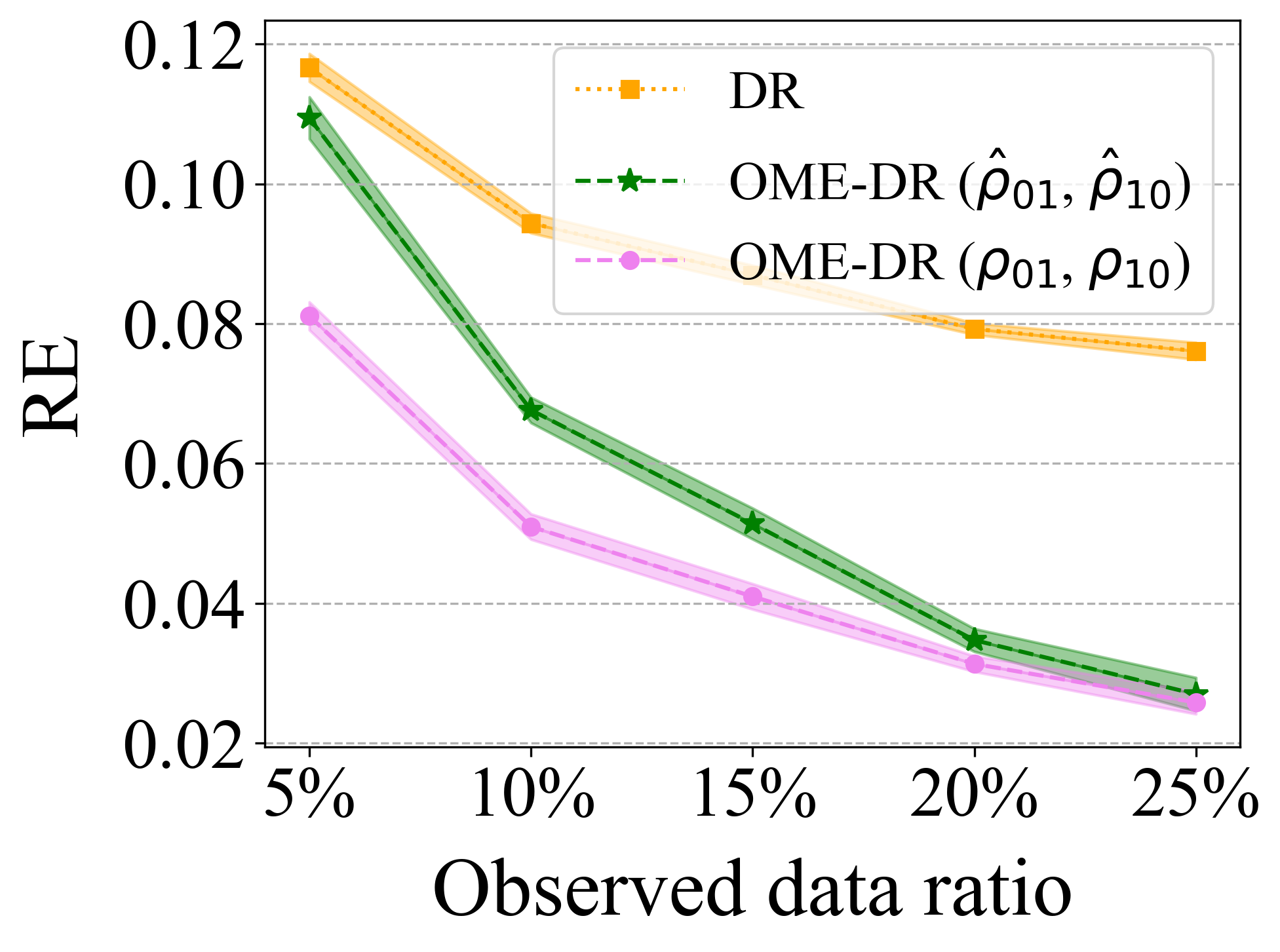}
\end{minipage}%
}}%
\centering
\vspace{-12pt}
\caption{Estimation errors of $\hat{\rho}_{01}$ and $\hat{\rho}_{10}$ and RE with varying observed data ratios.}
\vspace{-6pt}
\label{fig:semi}
\end{figure*}
\begin{table*}[]
\centering
\setlength{\tabcolsep}{3pt}
\caption{Performance on AUC, NDCG@K, and Recall@K on the MAR test set of \textsc{Coat}, \textsc{Music} and \textsc{KuaiRec} with $\rho_{01}$ = 0.2 and $\rho_{10}$ = 0.1. The best three results are bolded, and the best baseline is underlined.}
\vspace{-8pt}
\resizebox{0.945\linewidth}{!}{
\begin{tabular}{l|ccc | ccc | ccc}
\toprule
\multicolumn{1}{c|}{}    & \multicolumn{3}{c|}{\textsc{Coat}} & \multicolumn{3}{c|}{\textsc{Music}}      & \multicolumn{3}{c}{\textsc{KuaiRec}}         \\ \midrule
 Method & AUC     & NDCG@5        & Recall@5 & AUC     & NDCG@5        & Recall@5          &   AUC     & NDCG@50        & Recall@50  
\\ \midrule

MF~\cite{koren2009matrix} & 0.619$_{\pm0.002}$ & 0.531$_{\pm0.005}$ & 0.364$_{\pm0.005}$ & 0.632$_{\pm0.005}$ & 0.616$_{\pm0.003}$ & 0.361$_{\pm0.001}$ & 0.621$_{\pm0.002}$ & 0.652$_{\pm0.003}$ & 0.702$_{\pm0.002}$ \\ 
\midrule
OME~\cite{liu2015classification} & 0.621$_{\pm0.003}$ & 0.543$_{\pm0.006}$ & 0.371$_{\pm0.007}$ & 0.650$_{\pm0.005}$ & 0.615$_{\pm0.004}$ & 0.373$_{\pm0.004}$ & 0.645$_{\pm0.002}$ & 0.660$_{\pm0.002}$ & 0.722$_{\pm0.002}$ \\

T-MF~\cite{wang2021denoising} & 0.633$_{\pm0.004}$ & 0.540$_{\pm0.005}$ & 0.372$_{\pm0.006}$ & 0.650$_{\pm0.005}$ & 0.614$_{\pm0.003}$ & 0.376$_{\pm0.004}$ & 0.657$_{\pm0.002}$ & 0.665$_{\pm0.002}$ & 0.739$_{\pm0.003}$ \\

R-MF~\cite{wang2021denoising} & 0.641$_{\pm0.004}$ & 0.544$_{\pm0.005}$ & 0.376$_{\pm0.009}$ & 0.657$_{\pm0.005}$ & 0.614$_{\pm0.006}$ & 0.371$_{\pm0.004}$ & 0.645$_{\pm0.002}$ & 0.661$_{\pm0.001}$ & 0.720$_{\pm0.003}$ \\
LCD-MF~\cite{dai2022lcd} & 0.631$_{\pm0.003}$ & 0.552$_{\pm0.007}$ & \underline{\textbf{0.381}}$_{\pm0.006}$ & 0.659$_{\pm0.005}$ & 0.617$_{\pm0.005}$ & 0.374$_{\pm0.004}$ & 0.664$_{\pm0.003}$ & 0.677$_{\pm0.003}$ & 0.738$_{\pm0.003}$ \\
\midrule
EIB~\cite{Steck2010} & 0.621$_{\pm0.003}$ & 0.547$_{\pm0.006}$ & 0.368$_{\pm0.008}$ & 0.648$_{\pm0.003}$ & 0.620$_{\pm0.002}$ & 0.371$_{\pm0.002}$  & 0.641$_{\pm0.001}$ & 0.653$_{\pm0.001}$ & 0.713$_{\pm0.002}$\\

IPS~\cite{Schnabel-Swaminathan2016} & 0.625$_{\pm0.002}$ & 0.534$_{\pm0.009}$ & 0.360$_{\pm0.008}$ &  0.648$_{\pm0.001}$ & 0.615$_{\pm0.001}$ & 0.369$_{\pm0.001}$ & 0.636$_{\pm0.001}$ & 0.642$_{\pm0.002}$ & 0.714$_{\pm0.001}$ \\  

SNIPS~\cite{Schnabel-Swaminathan2016} & 0.631$_{\pm0.004}$ & 0.544$_{\pm0.008}$ & 0.369$_{\pm0.007}$ &  0.659$_{\pm0.005}$ & 0.617$_{\pm0.005}$ & 0.375$_{\pm0.004}$ & 0.657$_{\pm0.002}$ & 0.670$_{\pm0.001}$ & 0.735$_{\pm0.001}$   \\

CVIB~\cite{CVIB} & 0.626$_{\pm0.001}$ & 0.552$_{\pm0.009}$ & 0.379$_{\pm0.004}$ & 0.661$_{\pm0.001}$ & 0.612$_{\pm0.002}$ & 0.381$_{\pm0.001}$ & 0.653$_{\pm0.001}$ & 0.644$_{\pm0.001}$ & 0.736$_{\pm0.001}$  \\

DAMF~\cite{saito2019towards} & 0.631$_{\pm0.004}$ & 0.548$_{\pm0.006}$ & 0.367$_{\pm0.006}$ & 0.640$_{\pm0.003}$  & \textbf{\underline{0.627$_{\pm0.002}$}} & 0.377$_{\pm0.004}$ & 0.657$_{\pm0.003}$ & 0.669$_{\pm0.003}$ & 0.735$_{\pm0.004}$ \\

DR~\cite{saito2020doubly} & 0.633$_{\pm0.009}$ & 0.548$_{\pm0.015}$ & 0.360$_{\pm0.011}$ & 0.654$_{\pm0.009}$ & 0.613$_{\pm0.007}$ & 0.372$_{\pm0.005}$ & 0.643$_{\pm0.002}$ & 0.650$_{\pm0.001}$ & 0.721$_{\pm0.001}$ \\

DR-JL~\cite{Wang-Zhang-Sun-Qi2019} & 0.633$_{\pm0.005}$ & 0.546$_{\pm0.006}$ & 0.366$_{\pm0.007}$ & 0.650$_{\pm0.002}$ & 0.617$_{\pm0.001}$ & 0.382$_{\pm0.001}$  & 0.644$_{\pm0.001}$ & 0.669$_{\pm0.002}$ & 0.725$_{\pm0.002}$ \\

MRDR-JL~\cite{MRDR} & 0.634$_{\pm0.005}$ & 0.538$_{\pm0.010}$ & 0.373$_{\pm0.008}$ & 0.659$_{\pm0.002}$ & 0.624$_{\pm0.007}$ & 0.376$_{\pm0.006}$ & 0.658$_{\pm0.003}$ & 0.670$_{\pm0.002}$ & \underline{0.741$_{\pm0.003}$} \\

DR-BIAS~\cite{Dai-etal2022} & 0.632$_{\pm0.004}$ & 0.549$_{\pm0.009}$ &  0.369$_{\pm0.008}$ & 0.660$_{\pm0.001}$ & 0.619$_{\pm0.003}$ & 0.378$_{\pm0.002}$  & \textbf{\underline{0.665$_{\pm0.004}$}} & \underline{\textbf{0.679}$_{\pm0.004}$} & 0.737$_{\pm0.002}$ \\

DR-MSE~\cite{Dai-etal2022} & 0.631$_{\pm0.003}$ & \textbf{\underline{0.557$_{\pm0.008}$}} & 0.374$_{\pm0.008}$ & \underline{0.662$_{\pm0.002}$} & 0.623$_{\pm0.004}$ & 0.374$_{\pm0.002}$  & 0.653$_{\pm0.002}$ & 0.660$_{\pm0.002}$ & 0.732$_{\pm0.002}$ \\

DIB~\cite{liu2021mitigating} & 0.626$_{\pm0.003}$ & 0.550$_{\pm0.009}$ & 0.376$_{\pm0.006}$ & 0.647$_{\pm0.005}$ & \textbf{0.626}$_{\pm0.005}$ & 0.382$_{\pm0.006}$ & 0.641$_{\pm0.002}$ & 0.642$_{\pm0.003}$ & 0.738$_{\pm0.003}$ \\

MR~\cite{li-etal-MR2023} & 0.630$_{\pm0.005}$ & 0.546$_{\pm0.009}$ & 0.376$_{\pm0.007}$ & 0.656$_{\pm0.005}$ & 0.623$_{\pm0.004}$ & 0.380$_{\pm0.004}$ & 0.637$_{\pm0.002}$ & 0.644$_{\pm0.004}$ & 0.723$_{\pm0.002}$ \\

SDR~\cite{SDR} & 0.635$_{\pm0.005}$ & 0.543$_{\pm0.005}$ & 0.368$_{\pm0.006}$ &  0.660$_{\pm0.006}$ & 0.618$_{\pm0.006}$ & 0.376$_{\pm0.005}$ & 0.661$_{\pm0.003}$ & 0.675$_{\pm0.001}$ & 0.739$_{\pm0.002}$ \\

TDR~\cite{TDR} &  \textbf{\underline{0.638$_{\pm0.007}$}} & 0.556$_{\pm0.014}$ & 0.370$_{\pm0.013}$ &  0.657$_{\pm0.004}$ & 0.616$_{\pm0.002}$ & 0.371$_{\pm0.003}$  & 0.650$_{\pm0.002}$ & 0.670$_{\pm0.001}$ & 0.723$_{\pm0.003}$ \\

IPS-V2~\cite{li2023propensity} & 0.631$_{\pm0.005}$ & 0.549$_{\pm0.008}$ & 0.373$_{\pm0.008}$ & 0.648$_{\pm0.005}$ & 0.615$_{\pm0.005}$ & 0.373$_{\pm0.004}$  & 0.658$_{\pm0.002}$ & 0.675$_{\pm0.003}$ & 0.734$_{\pm0.002}$ \\

DR-V2~\cite{li2023propensity} & 0.631$_{\pm0.005}$ & 0.553$_{\pm0.009}$ & 0.378$_{\pm0.006}$ & 0.636$_{\pm0.003}$ & 0.620$_{\pm0.002}$ & \underline{\textbf{0.392}$_{\pm0.004}$} & 0.658$_{\pm0.001}$ & 0.671$_{\pm0.001}$ & 0.734$_{\pm0.002}$ \\

\midrule

OME-EIB (ours) & 0.627$_{\pm0.002}$ & \textbf{0.563$^*_{\pm0.002}$} & \textbf{0.387$^*_{\pm0.002}$}  &  \textbf{0.664$_{\pm0.003}$} & 0.622$_{\pm0.004}$ & 0.384$_{\pm0.005}$  & \textbf{0.667$_{\pm0.003}$} & \textbf{0.682$^*_{\pm0.002}$} & \textbf{0.742$_{\pm0.002}$} \\

OME-IPS (ours) & \textbf{0.636}$_{\pm0.002}$ & 0.548$_{\pm0.003}$ & 0.373$_{\pm0.004}$ &  \textbf{0.664}$_{\pm0.003}$ & 0.625$_{\pm0.002}$ & \textbf{0.385$_{\pm0.002}$} & \textbf{0.671$^*_{\pm0.002}$} & \textbf{0.693$^*_{\pm0.004}$} & \textbf{0.763$^*_{\pm0.002}$} \\

OME-DR (ours) & \textbf{0.651$^*_{\pm0.006}$} & \textbf{0.561$_{\pm0.006}$} & \textbf{0.385$^*_{\pm0.005}$} & \textbf{0.671$^*_{\pm0.003}$} & \textbf{0.632$^*_{\pm0.003}$} & \textbf{0.389$_{\pm0.002}$} & 0.662$_{\pm0.002}$ & 0.677$_{\pm0.004}$ & \textbf{0.743$_{\pm0.002}$} \\

\bottomrule
\end{tabular}
}%
\label{tab:real-world}%
\begin{tablenotes}
\scriptsize
\item Note: * means statistically significant results ($\text{p-value} \leq 0.05$) using the paired-t-test compared with the best baseline.
\end{tablenotes}
\vspace{-8pt}
\end{table*}
\vspace{-2pt}
\section{Semi-Synthetic Experiments}
To investigate if the proposed estimators are able to estimate the true prediction inaccuracy accurately in the presence of OME and MNAR effect, several semi-synthetic experiments are conducted on a widely-used dataset MovieLens-100K (\textsc{ML-100K}). We first adopt MF to complete the five-scaled rating matrix $\mathbf{R}$. However, the completed matrix will have an unrealistic rating distribution, thus we sort the matrix entries in ascending order and assign a positive feedback probability of 0.1 for the lowest $p_1$ proportion, a positive feedback probability of 0.3 for the next $p_2$ proportion, and so on to obtain a true positive feedback probability $\gamma_{u, i}$ for each user-item pair. The adjusted probability matrix contains $\gamma_{u, i} \in\{0.1, 0.3, 0.5, 0.7, 0.9\}$ with proportion $\left[p_{1}, p_{2}, p_{3}, p_{4}, p_{5}\right]$, respectively.

Second, following the previous studies~\cite{Wang-Zhang-Sun-Qi2019, Schnabel-Swaminathan2016, MRDR}, we use six matrices below as the prediction matrix $\hat{\mathbf{R}}$: \\
$\bullet$ \textbf{ROTATE}:  Set predicted $\hat{r}_{u, i}=\gamma_{u, i}-0.2$ when $\gamma_{u, i} \geq 0.3$, and $\hat{r}_{u, i}=0.9$ when $\gamma_{u, i}=0.1$.\\
$\bullet$ \textbf{SKEW}: Predicted $\hat{r}_{u, i}$ are sampled from the Gaussian distribution $\mathcal{N}(\mu=\gamma_{u, i}, \sigma= (1-\gamma_{u, i})/{2})$, and clipped to the interval $[0.1,0.9]$.\\
$\bullet$ \textbf{CRS}: If the true $\gamma_{u, i} \leq 0.6$, then $\hat{r}_{u, i}=0.2$. Otherwise, $\hat{r}_{u, i}=0.6$.\\
$\bullet$ \textbf{ONE}: The predicted matrix $\mathbf{\hat R}$ is identical to the true positive feedback probability matrix, except that randomly select $\gamma_{u,i}=0.1$ with total amount $|\{(u, i) \mid \gamma_{u, i}=0.9\}|$ are flipped to 0.9. \\
$\bullet$ \textbf{THREE}: Same as \textbf{ONE}, but flipping $\gamma_{u,i} = 0.3$ instead.\\
$\bullet$ \textbf{FIVE}: Same as \textbf{ONE}, but flipping $\gamma_{u,i} = 0.5$ instead.

Next, we assign the propensity $p_{u, i}=p \alpha^{\min(4,  6- r_{u, i})}$ for each user-item pair and obtain the estimate propensities $$\frac{1}{\hat{p}_{u, i}}=\frac{1-{\beta}}{p_{u, i}}+\frac{\beta}{p_{e}},$$ where $p_{e}= |\mathcal{D}|^{-1} \sum_{(u, i) \in \mathcal{D}} o_{u, i}$, $p$ is set to 1 in our experiment and {$\beta$} is randomly sampled from a uniform distribution ${U}(0,\,1)$ to introduce noises. Then, we sample the binary true feedback matrix $\mathbf{R}^*$ and binary observation matrix $\mathbf{O}$ as follows:
$$o_{u, i} \sim \operatorname{Bern}(p_{u, i}),  \forall(u, i) \in \mathcal{D}, \quad r^*_{u, i} \sim \operatorname{Bern}(\gamma_{u, i}),  \forall(u, i) \in \mathcal{D},$$
where $\operatorname{Bern}(\cdot)$ denotes
the Bernoulli distribution. Then we flip the feedback matrix according to the $\rho_{01}$ and $\rho_{10}$ to generate a binary noise feedback matrix $\mathbf{R}$. The absolute relative error (RE) is used for evaluation, which is defined as $\operatorname{RE} (\mathcal{E}_{est})= |\mathcal{P}^*-\mathcal{E}_{est}(\hat{\mathbf{R}}, \mathbf{R})| / \mathcal{P}^*$,  
where $\cE_{est}$ denotes the estimate prediction inaccuracy. The smaller the RE, the more accurate the estimation. In addition, we provide both results for using estimated $\hat{\rho}_{01}$ and $\hat{\rho}_{10}$ to estimate $\mathcal{P}^*$ and directly using true $\rho_{01}$ and $\rho_{10}$ to estimate $\mathcal{P}^*$.

\vspace{4pt}\noindent{{\bf Performance Comparison.}}
We compare the proposed methods with the Naive method~\cite{koren2009matrix}, the EIB method~\cite{Steck2010}, the IPS method~\cite{Schnabel-Swaminathan2016}, and the DR-based methods~\cite{Wang-Zhang-Sun-Qi2019,MRDR,TDR,SDR}. The RE results are shown in Table \ref{tab:semi1} with $\rho_{01} = 0.2$ and $\rho_{10} = 0.1$. First, most debiasing and denoising methods have lower RE than the Naive method. In addition, our proposed methods stably and significantly outperform the corresponding debiasing baseline methods. Meanwhile, the OME methods are only able to denoising, thus they still have a large RE, even the true $\rho$ is used for estimation. This shows the effectiveness of our method in the presence of both the OME and MNAR effects.

\vspace{2.7pt}\noindent{{\bf In Depth Analysis.}} We explore the effect of the proportion of observed data on the estimation accuracy of $\hat{\rho}_{01}$ and $\hat{\rho}_{10}$, and the results when $\rho_{01} = 0.2$ and $\rho_{10} = 0.1$ are shown in Figure \ref{fig:semi}. We take ROTATE as an example, and the same phenomenon is observed for the other five prediction matrices. First, our methods stably outperform the baseline methods in all scenarios. Second, the estimation error decreases significantly as the proportion of observed data increases. Moreover, even with the estimation error, our estimation results are only slightly worse than using the real $\rho$. This further validates the stability and practicality of our methods.
\section{Real-World Experiments}
\textbf{Datasets and Experiment Details.} We verify the effectiveness of our methods on three real-world datasets: \textsc{Coat}~\cite{Schnabel-Swaminathan2016}, \textsc{Music}~\cite{Schnabel-Swaminathan2016} and \textsc{KuaiRec}~\cite{gao2022kuairec}. \textsc{Coat} includes 6,960 
MNAR ratings and 4,640 missing-at-random (MAR) ratings from 290 users to 300 items. \textsc{Music} has 311,704 MNAR ratings and 54,000 MAR ratings of 15,400 users to 1,000 items. For \textsc{Coat} and \textsc{Music}, we binarize the ratings less than three to 0 and otherwise to 1. \textsc{KuaiRec} is a fully exposed large-scale industrial dataset, which has 4,676,570 video watching ratio records from 1,411 users to 3,327 items. We binarize the
records less than one to 0 and otherwise to 1. We adopt three common metrics, \emph{i.e.}, AUC, NDCG@K, and Recall@K for performance evaluation. For \textsc{Coat} and \textsc{Music}, K is set to 5. For \textsc{KuaiRec}, K is set to 50. All the experiments are implemented on PyTorch with Adam as the optimizer. For all experiments, we use GeForce RTX 3090 as the computing resource. Logistic regression is used as the propensity model for all the methods requiring propensity. The learning rate is tuned in $\{0.001, 0.005, 0.01, 0.05\}$ and the batch size is tuned in $\{64, 128, 256\}$ for \textsc{Coat} and $\{2048, 4096, 8192\}$ for \textsc{Music} and \textsc{KuaiRec}. We tune the embedding dimension in $\{4, 8, 16, 32\}$ for \textsc{Coat} and $\{8, 16, 32, 64\}$ for \textsc{Music} and \textsc{KuaiRec}. Moreover, we tune the weight decay rate in $[1e-6, 5e-3]$.

\vspace{4pt}\noindent{{\bf Baselines.}} We use matrix factorization (MF)~\cite{koren2009matrix} as the base model and compare our methods to the following debiasing methods: EIB~\cite{Steck2010}, IPS~\cite{Schnabel-Swaminathan2016}, SNIPS~\cite{Schnabel-Swaminathan2016}, CVIB~\cite{CVIB}, DAMF~\cite{saito2019towards}, DR~\cite{saito2020doubly}, DR-JL~\cite{Wang-Zhang-Sun-Qi2019}, MRDR-JL~\cite{MRDR}, DR-BIAS~\cite{Dai-etal2022}, DR-MSE~\cite{Dai-etal2022}, DIB~\cite{liu2021mitigating}, MR~\cite{li-etal-MR2023}, SDR~\cite{SDR}, TDR~\cite{TDR}, IPS-V2~\cite{li2023propensity} and DR-V2~\cite{li2023propensity}. We also compare denoising methods including surrogate loss minimization (OME)~\cite{liu2015classification}, T-MF~\cite{wang2021denoising}, R-MF~\cite{wang2021denoising}, and LCD-MF~\cite{dai2022lcd}.


\vspace{4pt}\noindent{{\bf Performance Comparision.}} Table \ref{tab:real-world} shows the prediction performance with varying baselines and our methods. First, most of the debiasing and denoising methods have better performance compared to the Naive method, which shows the necessity of debiasing and denoising. Meanwhile, our methods exhibit the most competitive performance in all three datasets, significantly outperforming the baselines including debiasing methods and denoising methods.

\vspace{-4pt}
\section{Related Work}
\subsection{Debiased Recommendation}
The data collected in recommender systems are often systematically subject to varying types of bias, such as conformity bias~\cite{DBLP:conf/recsys/LiuCY16}, item popularity bias~\cite{wei2021model,zhang2021causal}, latent confounders~\cite{Ding-etal2022,li2023balancing}, and position bias~\cite{oosterhuis2022reaching,oosterhuis2023doubly}. In order to achieve unbiased learning of the prediction model, three typical approaches have been proposed, including: (1) The error-imputation-based (EIB) approaches~\cite{Steck2010,Lobato-etal2014,saito2020asymmetric}, which compute an imputed error for each missing rating. (2) The inverse-propensity-scoring (IPS) approaches~\cite{Schnabel-Swaminathan2016,saito2020ips,swaminathan2015self,saito2019unbiased}, which inversely weight the prediction error for each observed rating with the probability of observing that rating. (3) The doubly robust (DR) approaches~\cite{Wang-Zhang-Sun-Qi2019,saito2020doubly,kweon2024doubly}, which use both the error imputation model and the propensity model, and the unbiased learning of the prediction model can be achieved when either the error imputation model or the propensity model is accurate. Based on the above EIB, IPS, and DR estimators, in terms of learning paradigms, recent studies have investigated flexible trade-offs between bias and variance~\cite{MRDR,Dai-etal2022}, parameter sharing in multi-task learning~\cite{ESMM,Multi_IPW,wang2022escm2}, and the use of a few unbiased ratings to improve the estimation of the prediction inaccuracy~\cite{bonner2018causal,Chen-etal2021,liu2020general,Wang-etal2021,liu2022kdcrec,lin2023transfer}. In terms of statistical theory, recent studies have developed targeted DR (TDR)~\cite{TDR} and conservative DR (CDR)~\cite{song2023cdr} to combat the inaccurate imputed errors, StableDR~\cite{SDR} to combat sparse data, and new propensity estimation methods based on balancing metrics~\cite{luo2021unbiased,li2023propensity,yang2023debiased}. In addition, \cite{saito2019towards} proposes to minimize the propensity-independent generalization error bound via adversarial learning. \cite{li-etal-MR2023} extends DR to multiple robust learning. These methods have also been applied to sequential recommendation~\cite{wang2022unbiased} and social recommendation~\cite{chen2018social}. Furthermore, methods based on information bottlenecks~\cite{liu2021mitigating,CVIB} and representation learning~\cite{yang2023debiased,pandiscriminative} have also been proposed for debiased recommendation. In this work, we extend the previous debiasing methods to a more realistic RS scenario, in which the observed ratings and the users' true preferences may be different.


\vspace{-2pt}
\subsection{Outcome Measurement Error}
Outcome measurement error (OME) refers to the inconsistency between the observed outcomes and the true outcomes, also known
as noisy outcomes in social science~\cite{roberts1985measurement}, biostatistics~\cite{hui1980estimating}, and 
psychometrics~\cite{Patrick2012}. The error models or transition matrices in OME establish the relation between the true outcomes and the observed outcomes, including uniform \cite{angluin1988learning,van2015learning},
class-conditional \cite{menon2015learning,scott2013classification,liu2015classification}, and instance-dependent \cite{chen2021beyond,xia2020part} structures
of outcome misclassification. Many data-driven error parameter estimation methods are developed in recent machine learning literature~\cite{northcutt2021confident,scott2013classification,scott2015rate,xia2019anchor}. Meanwhile, given knowledge
of measurement error parameters, recent studies propose unbiased risk minimization approaches for learning under noisy labels~\cite{natarajan2013learning,chou2020unbiased,patrini2017making,van2015machine}.

Despite the prevalence of OME in real-world recommendation scenarios, there is still only limited work focusing on denoising in RS~\cite{zhang2023robust,chen2022denoising,zhang2023denoising}. Recently, by noticing that noisy feedback typically has large loss values in the early stages, \cite{wang2021denoising} proposes adaptive denoising training and \cite{gao2022self} proposes self-guided denoising learning in implicit feedback. However, the existing methods are mostly heuristic and require the fully observed noisy labels, which prevents the direct use of such methods for RS in the presence of missing data. To fill this gap, we extend the surrogate loss-based methods to address the data MNAR and OME in parallel.

\vspace{-2pt}
\section{Conclusion}
In this study, we explored the challenges characterized MNAR data with noisy feedback encountered in real-world recommendation scenarios. First, we introduced the concept of true prediction inaccuracy as a more comprehensive evaluation criterion that accounts for OME, and proposed OME-EIB, OME-IPS, and OME-DR estimators to provide unbiased estimations of the true prediction inaccuracy. Next, we derived the explicit forms of the biases and the generalization bounds of the proposed estimators, and proved the double robustness of our OME-DR estimator. We further proposed an alternating denoise training approach that corrects for MNAR data with OME. The effectiveness of our methods was validated on both semi-synthetic and real-world datasets with varying MNAR and OME rates. The potential limitations includes the usage of weak separability assumption, and accurate measurement error parameters' estimation for ensuring unbiasedness. In future research, we aim to develop unbiased learning methods applicable to more complex OME forms, 
study alternative practical assumptions for identifying and estimating the error parameters, and extend the applicability of the proposed methods to a wider range of RS settings.
\vspace{-2pt}
\section*{Acknowledgements}
This work was supported in part by National Natural Science Foundation of China (623B2002, 62272437).

\newpage
\bibliographystyle{ACM-Reference-Format}
\bibliography{reference}

\newpage
\appendix
\section{Proofs}\label{app-a}
\textsc{Theorem 3.1 (Bias of OME-DR Estimator).} \emph{Given $\hat \rho_{01}$ and  $\hat \rho_{10}$ with $\hat \rho_{01}+\hat \rho_{10}<1$, imputed errors $\bar{\mathbf{E}}$ and learned propensities $\hat{\mathbf{P}}$ with $\hat{p}_{u, i}>0$ for all user-item pairs, the bias of the OME-DR estimator is
\begin{align*}
&\operatorname{Bias}[\mathcal{E}_{\mathrm{OME-DR}}(\hat{\mathbf{R}}, \mathbf{R}^o; \hat \rho_{01}, \hat \rho_{10})]=
\frac{1}{|\cD|}\left|\sum_{(u,i) \in \cD}\left(1-\frac{o_{u, i}}{\hat p_{u, i}}\right)\bar e_{u, i}\right.+\\
&\sum_{(u,i): r^*_{u,i}=1}\left(\frac{p_{u, i}\omega_{11}-\hat p_{u, i}}{\hat p_{u, i}}\ell(f_{\theta}(x_{u,i}),1)+\frac{p_{u, i}\omega_{01}}{\hat p_{u, i}}\ell(f_{\theta}(x_{u,i}),0)\right)+\\
&\left.\sum_{(u,i): r^*_{u,i}=0}\left(\frac{p_{u, i}\omega_{10}}{\hat p_{u, i}}\ell(f_{\theta}(x_{u,i}),1)+\frac{p_{u, i}\omega_{00}-\hat p_{u, i}}{\hat p_{u, i}}\ell(f_{\theta}(x_{u,i}),0)\right)\right|,
\end{align*}
where $\omega_{11}$, $\omega_{01}$, $\omega_{10}$, and $\omega_{00}$ are given by
\begin{align*}
\omega_{11}&=\frac{1-\rho_{01}-\hat \rho_{10}}{1-\hat \rho_{01}-\hat \rho_{10}}, \quad \omega_{01}= \frac{\rho_{01}-\hat \rho_{01}}{1-\hat \rho_{01}-\hat \rho_{10}},\\
\omega_{10}&=\frac{\rho_{10}-\hat \rho_{10}}{1-\hat \rho_{01}-\hat \rho_{10}}, \quad \omega_{00}=\frac{1-\hat \rho_{01}- \rho_{10}}{1-\hat \rho_{01}-\hat \rho_{10}}.
\end{align*}}
\begin{proof}
By definition, the bias of the OME-DR estimator is
\begin{align*}
\operatorname{Bias}[\mathcal{E}_{\mathrm{OME-DR}}(\hat{\mathbf{R}}, \mathbf{R}^o; \hat \rho_{01}, \hat \rho_{10})]=\left|\mathbb{E}_{\mathbf{R},\mathbf{O}}\left[\mathcal{E}_{\mathrm{OME-DR}}\right]-\mathcal{P}^*\right|.
\end{align*}
By double expectation formula, the first term is
\begin{align*}
&\mathbb{E}_{\mathbf{R},\mathbf{O}}\left[\mathcal{E}_{\mathrm{OME-DR}}(\hat{\mathbf{R}}, \mathbf{R}^o; \hat \rho_{01}, \hat \rho_{10})\right]\\
={}&{}\mathbb{E}_{\mathbf{R}\mid \mathbf{O}}\left[\mathbb{E}_{\mathbf{O}}\left[\mathcal{E}_{\mathrm{OME-DR}}(\hat{\mathbf{R}}, \mathbf{R}^o; \hat \rho_{01}, \hat \rho_{10})\right]\right]\\
={}&{}\mathbb{E}_{\mathbf{R}\mid \mathbf{O}}\left[\frac{1}{|\cD|} \sum_{(u,i) \in \cD} \left(\left(1-\frac{p_{u, i}}{\hat p_{u, i}}\right)\bar e_{u, i}+\frac{p_{u, i}\tilde e_{u, i}}{\hat p_{u, i}}\right)\right]\\
={}&{}\frac{1}{|\cD|} \sum_{(u,i): r^*_{u,i}=1} \mathbb{E}_{\mathbf{R}\mid \mathbf{O}}\left[\left(1-\frac{p_{u, i}}{\hat p_{u, i}}\right)\bar e_{u, i}+\frac{p_{u, i}\tilde e_{u, i}}{\hat p_{u, i}}\right]\\
+{}&{}\frac{1}{|\cD|}\sum_{(u,i): r^*_{u,i}=0}\mathbb{E}_{\mathbf{R}\mid \mathbf{O}}\left[\left(1-\frac{p_{u, i}}{\hat p_{u, i}}\right)\bar e_{u, i}+\frac{p_{u, i}\tilde e_{u, i}}{\hat p_{u, i}}\right].
\end{align*}
Then, we can derive the bias of the OME-DR estimator as follows
\begin{align*}
&\operatorname{Bias}[\mathcal{E}_{\mathrm{OME-DR}}(\hat{\mathbf{R}}, \mathbf{R}^o; \hat \rho_{01}, \hat \rho_{10})]=\left|\mathbb{E}_{\mathbf{R}\mid \mathbf{O}}\left[\mathbb{E}_{\mathbf{O}}\left[\mathcal{E}_{\mathrm{OME-DR}}\right]\right]-\mathcal{P}^*\right|\\
={}&{}\left|\frac{1}{|\cD|} \sum_{(u,i): r^*_{u,i}=1} \mathbb{E}_{\mathbf{R}\mid \mathbf{O}}\left[\left(1-\frac{p_{u, i}}{\hat p_{u, i}}\right)\bar e_{u, i}+\frac{p_{u, i}\tilde e_{u, i}}{\hat p_{u, i}}\right]-\ell(f_{\theta}(x_{u,i}),1)\right.\\
+{}&{}\left.\frac{1}{|\cD|}\sum_{(u,i): r^*_{u,i}=0}\mathbb{E}_{\mathbf{R}\mid \mathbf{O}}\left[\left(1-\frac{p_{u, i}}{\hat p_{u, i}}\right)\bar e_{u, i}+\frac{p_{u, i}\tilde e_{u, i}}{\hat p_{u, i}}\right]-\ell(f_{\theta}(x_{u,i}),0)\right|\\
={}&{}\frac{1}{|\cD|}\left|\sum_{(u,i): r^*_{u,i}=1}\left(1-\frac{p_{u, i}}{\hat p_{u, i}}\right)\bar e_{u, i}+\sum_{(u,i): r^*_{u,i}=0}\left(1-\frac{p_{u, i}}{\hat p_{u, i}}\right)\bar e_{u, i}\right.\\
+{}&{}\sum_{(u,i): r^*_{u,i}=1}\left(\frac{p_{u, i}}{\hat p_{u, i}}\cdot\mathbb{E}_{\mathbf{R}\mid \mathbf{O}}\left[\tilde e_{u,i}\right]-\ell(f_{\theta}(x_{u,i}),1)\right)\\
+{}&{}\left.\sum_{(u,i): r^*_{u,i}=0}\left(\frac{p_{u, i}}{\hat p_{u, i}}\cdot\mathbb{E}_{\mathbf{R}\mid \mathbf{O}}\left[\tilde e_{u,i}\right]-\ell(f_{\theta}(x_{u,i}),0)\right)\right|.
\end{align*}
On one hand, for the user-item pairs with $r^*_{u,i}=1$, we have
\begin{align*}
&\mathbb{E}_{\mathbf{R}\mid \mathbf{O}}\left[\tilde e_{u,i}\right]\\
={}&{}\left(1-\rho_{01}\right) \cdot\tilde{\ell}(f_{\theta}(x_{u,i}),1; \hat \rho_{01}, \hat \rho_{10})+\rho_{01}\cdot \tilde{\ell}(f_{\theta}(x_{u,i}),0; \hat \rho_{01}, \hat \rho_{10})\\
={}&{}\left(1-\rho_{01}\right) \cdot \frac{\left(1-\hat \rho_{10}\right) \ell(f_{\theta}(x_{u,i}),1)-\hat \rho_{01} \ell(f_{\theta}(x_{u,i}),0)}{1-\hat \rho_{01}-\hat \rho_{10}}\\
&{}+\rho_{01}\cdot\frac{\left(1-\hat \rho_{01}\right) \ell(f_{\theta}(x_{u,i}),0)-\hat \rho_{10} \ell(f_{\theta}(x_{u,i}),1)}{1-\hat \rho_{01}-\hat \rho_{10}}\\
={}&{}\frac{1-\rho_{01}-\hat \rho_{10}}{1-\hat \rho_{01}-\hat \rho_{10}}\ell(f_{\theta}(x_{u,i}),1)+\frac{\rho_{01}-\hat \rho_{01}}{1-\hat \rho_{01}-\hat \rho_{10}}\ell(f_{\theta}(x_{u,i}),0)\\
:={}&{}\omega_{11}\ell(f_{\theta}(x_{u,i}),1)+\omega_{01}\ell(f_{\theta}(x_{u,i}),0),
\end{align*}
where $\omega_{11}$ and $\omega_{01}$ are given by
\begin{align*}
\omega_{11}=\frac{1-\rho_{01}-\hat \rho_{10}}{1-\hat \rho_{01}-\hat \rho_{10}}, \quad \omega_{01}= \frac{\rho_{01}-\hat \rho_{01}}{1-\hat \rho_{01}-\hat \rho_{10}}.
\end{align*}
On the other hand, for the user-item pairs with $r^*_{u,i}=0$, we have
\begin{align*}
&\mathbb{E}_{\mathbf{R}\mid \mathbf{O}}\left[\tilde e_{u,i}\right]\\
={}&{} \rho_{10} \cdot\tilde{\ell}(f_{\theta}(x_{u,i}),1; \hat \rho_{01}, \hat \rho_{10})+(1-\rho_{10})\cdot \tilde{\ell}(f_{\theta}(x_{u,i}),0; \hat \rho_{01}, \hat \rho_{10})\\
={}&{}\rho_{10} \cdot \frac{\left(1-\hat \rho_{10}\right) \ell(f_{\theta}(x_{u,i}),1)-\hat \rho_{01} \ell(f_{\theta}(x_{u,i}),0)}{1-\hat \rho_{01}-\hat \rho_{10}}\\
&{}+(1-\rho_{10})\cdot\frac{\left(1-\hat \rho_{01}\right) \ell(f_{\theta}(x_{u,i}),0)-\hat \rho_{10} \ell(f_{\theta}(x_{u,i}),1)}{1-\hat \rho_{01}-\hat \rho_{10}}\\
={}&{}\frac{\rho_{10}-\hat \rho_{10}}{1-\hat \rho_{01}-\hat \rho_{10}}\ell(f_{\theta}(x_{u,i}),1)+\frac{1-\hat \rho_{01}- \rho_{10}}{1-\hat \rho_{01}-\hat \rho_{10}}\ell(f_{\theta}(x_{u,i}),0)\\
:={}&{}\omega_{10}\ell(f_{\theta}(x_{u,i}),1)+\omega_{00}\ell(f_{\theta}(x_{u,i}),0),
\end{align*}
where $\omega_{10}$ and $\omega_{00}$ are given by
\begin{align*}
\omega_{10}=\frac{\rho_{10}-\hat \rho_{10}}{1-\hat \rho_{01}-\hat \rho_{10}}, \quad \omega_{00}=\frac{1-\hat \rho_{01}- \rho_{10}}{1-\hat \rho_{01}-\hat \rho_{10}}.
\end{align*}
This completes the proof.
\end{proof}



\textsc{Lemma 3.3 (Lipschitz Continuity).} \emph{Given $\hat \rho_{01}$ and  $\hat \rho_{10}$ with $\hat \rho_{01}+\hat \rho_{10}<1$, if $\ell(f_{\theta}(x_{u,i}), r^*_{u, i})$ is $L$-Lipschitz in $f_{\theta}(x_{u,i})$ for all $r^*_{u, i}$, then $\tilde{\ell}(f_{\theta}(x_{u,i}),r_{u, i})$ is $\frac{2L}{1-\hat \rho_{01}-\hat \rho_{10}}$-Lipschitz in $f_{\theta}(x_{u,i})$ for all $r_{u, i}$.}
\begin{proof}
The explicit form of $\tilde{\ell}(f_{\theta}(x_{u,i}),1)$ and $\tilde{\ell}(f_{\theta}(x_{u,i}),0)$ are
$$
\begin{aligned}
\tilde{\ell}(f_{\theta}(x_{u,i}),1) & =\frac{\left(1-\hat \rho_{10}\right) \ell(f_{\theta}(x_{u,i}),1)-\hat \rho_{01} \ell(f_{\theta}(x_{u,i}),0)}{1-\hat \rho_{01}-\hat \rho_{10}}, \\
\tilde{\ell}(f_{\theta}(x_{u,i}),0) & =\frac{\left(1-\hat \rho_{01}\right) \ell(f_{\theta}(x_{u,i}),0)-\hat \rho_{10} \ell(f_{\theta}(x_{u,i}),1)}{1-\hat \rho_{01}-\hat \rho_{10}}.
\end{aligned}
$$
Then the conclusion can be dirived directly from $\hat \rho_{01}+\hat \rho_{10}<1$.
\end{proof}

\begin{lemma}[McDiarmid's Inequality]
Let $V$ be some set and let $f: V^m \rightarrow \mathbb{R}$ be a function of $m$ variables such that for some $c>0$, for all $i \in[m]$ and for all $x_1, \ldots, x_m, x_i^{\prime} \in V$ we have
$$
\left|f\left(x_1, \ldots, x_m\right)-f\left(x_1, \ldots, x_{i-1}, x_i^{\prime}, x_{i+1}, \ldots, x_m\right)\right| \leq c
$$
Let $X_1, \ldots, X_m$ be $m$ independent random variables taking values in $V$. Then, with probability of at least $1-\eta$ we have
$$
\left|f\left(X_1, \ldots, X_m\right)-\mathbb{E}\left[f\left(X_1, \ldots, X_m\right)\right]\right| \leq c \sqrt{\log \left(\frac{2}{\eta}\right) m / 2}.
$$    
\end{lemma}
\begin{proof}
The proof can be found in Lemma 26.4 of~\citep{shalev2014understanding}.
\end{proof}

\textsc{Theorem 3.4 (Generalization Bound).} \emph{Given $\hat \rho_{01}$ and  $\hat \rho_{10}$ with $\hat \rho_{01}+\hat \rho_{10}<1$, suppose $\ell(f_{\theta}(x_{u,i}), r^*_{u, i})$ is $L$-Lipschitz in $f_{\theta}(x_{u,i})$ for all $r^*_{u, i}$, and $\hat p_{u, i}\geq C_p$, $|\tilde{\ell}(f_{\theta}(x_{u,i}),r_{u, i})|\leq C_l$ for all $r_{u, i}$, then with probability $1-\eta$, the true prediction inaccuracy $\mathcal{P}(\hat{\mathbf{R}}^\ddagger, \mathbf{R}^*)$ of the optimal prediction matrix using the OME-DR estimator with imputed errors $\bar{\mathbf{E}}$ and learned propensities $\hat{\mathbf{P}}$ has the upper bound
\begin{align*}
&\mathcal{E}_{\mathrm{OME-DR}}(\hat{\mathbf{R}}^\ddagger, \mathbf{R}^o; \hat \rho_{01}, \hat \rho_{10})+\operatorname{Bias}[\mathcal{E}_{\mathrm{OME-DR}}(\hat{\mathbf{R}}^\ddagger, \mathbf{R}^o; \hat \rho_{01}, \hat \rho_{10})]+\\
&\left(1+\frac{2}{C_p}\right)\left[\frac{4L}{1-\hat \rho_{01}-\hat \rho_{10}}\mathcal{R}(\mathcal{F})+\left(C_l+\frac{4L}{1-\hat \rho_{01}-\hat \rho_{10}}\right)\sqrt{\frac{2\log (4 / \eta)}{|\mathcal{D}|}}\right].
\end{align*}}

\begin{proof}
The true prediction inaccuracy can be decomposed as
\begin{align*}
&\mathcal{P}^*=\cE_\mathrm{OME-DR}+\left(\mathcal{P}^*-\bfE[\cE_\mathrm{OME-DR}]\right)+\left(\bfE[\cE_\mathrm{OME-DR}]-\cE_\mathrm{OME-DR}\right)\\
&\leq{}\cE_\mathrm{OME-DR}+\operatorname{Bias}(\cE_\mathrm{OME-DR})+\left(\bfE[\cE_\mathrm{OME-DR}]-\cE_\mathrm{OME-DR}\right)\\
&\leq{} \cE_\mathrm{OME-DR}+\operatorname{Bias}(\cE_\mathrm{OME-DR})+ \sup _{f_\theta \in \mathcal{F}}\left(\bfE[\cE_\mathrm{OME-DR}]-\cE_\mathrm{OME-DR}\right).
\end{align*}
To simplify notation, let $\mathcal{B(\mathcal{F})}=\sup _{f_\theta \in \mathcal{F}}\left(\bfE[\cE_\mathrm{OME-DR}]-\cE_\mathrm{OME-DR}\right)$, which can be decomposed as 
\begin{align*}
\mathcal{B(\mathcal{F})}=\underset{S \sim \mathcal{\P}^{|\mathcal{D}|}}{\mathbb{E}}[\mathcal{B(\mathcal{F})}]+\left\{\mathcal{B(\mathcal{F})}-\underset{S \sim \mathcal{\P}^{|\mathcal{D}|}}{\mathbb{E}}[\mathcal{B(\mathcal{F})}]\right\}.
\end{align*}

For the first term $\underset{S \sim \mathcal{\P}^{|\mathcal{D}|}}{\mathbb{E}}[\mathcal{B(\mathcal{F})}]$, according to the Rademacher Comparison Lemma~\citep{shalev2014understanding} and Talagrand's Lemma~\citep{mohri2018foundations}, we have
\begin{align*}
&\underset{S \sim \mathcal{\P}^{|\mathcal{D}|}}{\mathbb{E}}[\mathcal{B(\mathcal{F})}]\\
&\leq 2 \underset{S \sim \mathcal{\P}^{|\mathcal{D}|}}{\mathbb{E}} \mathbb{E}_{\mathbf{\sigma}} \sup _{f_\theta \in \mathcal{F}}\left[\frac{1}{|\cD|}\sum_{(u, i)\in \cD}\sigma_{u, i}\left(\left(1-\frac{o_{u, i}}{\hat p_{u, i}}\right)\bar e_{u, i}+\frac{o_{u, i}\tilde e_{u, i}}{\hat p_{u, i}}\right)\right]\\
&\leq 2 \left(1+\frac{2}{C_p}\right)\frac{2L}{1-\hat \rho_{01}-\hat \rho_{10}}\cdot \underset{S \sim \mathcal{\P}^{|\mathcal{D}|}}{\mathbb{E}} \mathbb{E}_{\mathbf{\sigma}} \sup _{f_\theta \in \mathcal{F}}\left[\frac{1}{|\cD|}\sum_{(u, i)\in \cD}\sigma_{u, i}f_\theta(x_{u, i})\right]\\
&= 2 \left(1+\frac{2}{C_p}\right)\frac{2L}{1-\hat \rho_{01}-\hat \rho_{10}}\cdot \underset{S \sim \mathcal{\P}^{|\mathcal{D}|}}{\mathbb{E}}\{\mathcal{R}(\mathcal{F})\},
\end{align*}
where $\mathcal{R}(\mathcal{F})$ is the empirical Rademacher complexity of $\mathcal{F}$
\[
\mathcal{R}(\mathcal{F})=\mathbb{E}_{\mathbf{\sigma} \sim\{-1,+1\}^{|\mathcal{D}|}} \sup _{f_\theta \in \mathcal{F}}\left[\frac{1}{|\mathcal{D}|} \sum_{(u, i) \in \mathcal{D}} \sigma_{u, i} f_{\theta}(x_{u, i})\right].
\]

By applying McDiarmid's inequality, and let $c=\frac{2}{|\cD|}$,  with  probability at least $1-\frac{\eta}{2}$, 
$$
\left|\mathcal{R}(\mathcal{F})-\underset{S \sim \mathcal{\P}^{|\mathcal{D}|}}{\mathbb{E}}\{\mathcal{R}(\mathcal{F})\}\right| \leq 2 \sqrt{\frac{ \log (4 / \eta)}{2|\mathcal{D}|}}= \sqrt{\frac{2\log (4 / \eta)}{|\mathcal{D}|}}.
$$

For the rest term $\mathcal{B(\mathcal{F})}-\underset{S \sim \mathcal{\P}^{|\mathcal{D}|}}{\mathbb{E}}[\mathcal{B(\mathcal{F})}]$, by applying McDiarmid's inequality, and let $c=\frac{2}{|\cD|}C_l\left(1+\frac{2}{C_p}\right)$,  with  probability at least $1-\frac{\eta}{2}$, 
$$
\left|\mathcal{B(\mathcal{F})}-\underset{S \sim \mathcal{\P}^{|\mathcal{D}|}}{\mathbb{E}}[\mathcal{B(\mathcal{F})}]\right| \leq C_l\left(1+\frac{2}{C_p}\right) \sqrt{\frac{2\log (4 / \eta)}{|\mathcal{D}|}}.
$$
After adding the inequalities above, we can
rearrange the terms to obtain the stated results.
\end{proof}

\end{document}